\documentclass{new_tlp}

\usepackage{hyperref}
\usepackage{fixfoot}
\usepackage{amsmath}
\usepackage{stmaryrd}
\usepackage{latexsym}
\hypersetup{breaklinks}

\newtheorem{definition}{Definition} 
\newtheorem{example}{Example} 
\newtheorem{lemma}{Lemma}

\newtheorem{theorem}{Theorem}
\newtheorem{corollary}{Corollary}

\makeatletter
\newenvironment{oneshot}[1]{\@begintheorem{#1}{\unskip}}{\@endtheorem}
\makeatother

\DeclareFixedFootnote{\repnote}{See the Appendix for the relevant definitions.}

\newcommand{\hide}[1]{}

\newcommand*{\pnot}{\mathord{\sim}}

\newcommand{\mo}[1]{\llbracket#1\rrbracket}

\newcommand{\mwrs}[3]{\llbracket#1\rrbracket_{#2, #3}}

\newcommand{\vgr}[2]{#2(#1)}

\newcommand{\vwrt}[2]{v_{#2}(#1)}

\newcommand{\wfmp}{\mathcal{M}_\mathsf{P}}
\newcommand{\wgp}{M_\mathsf{Gr(P)}}
\newcommand{\pmp}{\mathcal{N}_\mathsf{P}}
\newcommand{\pgp}{N_\mathsf{Gr(P)}}

\newcommand{\tleq}[1][]{\leq_{#1}}
\newcommand{\fleq}[1][]{\preceq_{#1}}
\newcommand{\exeq}[1][]{\cong_{#1}}

\newcommand{\unk}{0}

\begin{document}
\bibliographystyle{acmtrans}

\long\def\comment#1{}

\title[The Intricacies of Three-Valued Extensional Semantics for Higher-Order LP]{The Intricacies of Three-Valued Extensional Semantics for Higher-Order Logic Programs}

  \author[P. Rondogiannis and I. Symeonidou]
		 {PANOS RONDOGIANNIS and IOANNA SYMEONIDOU \\
		 National and Kapodistrian University of Athens, Athens, Greece \\
		 \email{\{prondo,i.symeonidou\}@di.uoa.gr}}

\pagerange{\pageref{firstpage}--\pageref{lastpage}}
\volume{\textbf{10} (3):}
\jdate{March 2002}
\setcounter{page}{1}
\pubyear{2002}

\maketitle

\begin{abstract}
In~\cite{Bezem99,Bezem01}, M. Bezem defined an extensional semantics for
positive higher-order logic programs. Recently, it was demonstrated in~\cite{RS16}
that Bezem's technique can be extended to higher-order logic programs with negation,
retaining its extensional properties, provided that it is interpreted under
a logic with an infinite number of truth values. In~\cite{RS17} it was also
demonstrated that Bezem's technique, when extended under the stable model
semantics, does not in general lead to extensional stable models. In this paper
we consider the problem of extending Bezem's technique under the well-founded semantics.
We demonstrate that the well-founded extension {\em fails} to retain extensionality
in the general case. On the positive side, we demonstrate that for stratified
higher-order logic programs, extensionality is indeed achieved. We analyze the
reasons of the failure of extensionality in the general case, arguing that a
three-valued setting can not distinguish between certain predicates that appear
to have a different behaviour inside a program context, but which happen to be
identical as three-valued relations. The paper is under consideration for acceptance in TPLP.
\end{abstract}


\section{Introduction}
Recent research~\cite{Wa91a,Bezem99,Bezem01,KRW05,CharalambidisHRW13,RS16,ChRS17} has
investigated the possibility of providing {\em extensional} semantics to higher-order
logic programming. Under an extensional semantics, predicates denote sets, and therefore
one can use standard set theory in order to understand programs and reason
about them. Of course, extensionality comes with a price: to obtain an extensional
semantics, one usually has to consider higher-order logic programs with a
relatively restricted syntax. Actually this is a main difference between the extensional
and the more traditional {\em intensional} approaches to  higher-order logic programming
such as~\cite{MN2012,CKW93-187}: the latter languages have a richer syntax and expressive capabilities
but they are not usually amenable to a standard set-theoretic semantics.

There exist two main research directions for providing extensional semantics to higher-order
logic programs. The first one~\cite{Wa91a,KRW05,CharalambidisHRW13,CharalambidisER14} has
been developed using domain-theoretic tools, and resembles the techniques
for assigning denotational semantics to functional languages.
The second approach~\cite{Bezem99,Bezem01,RS16} relies on the syntactic entities that
exist in a program, and is based on processing the ground instantiation of the program.
The two research directions are not unrelated: it has recently been
shown~\cite{ChRS17} that for a broad class of positive programs, the two approaches
coincide with respect to ground atoms.

In this paper we focus exclusively on the second extensional approach. This approach was
initially proposed in~\cite{Bezem99,Bezem01} for {\em positive} (i.e., negationless) higher-order
logic programs. Recently, it was demonstrated in~\cite{RS16} that by combining the technique
of~\cite{Bezem99,Bezem01} with the infinite-valued semantics of~\cite{RondogiannisW05},
we obtain an extensional semantics for higher-order logic programs
with negation. In the extended version~\cite{RS17} of the paper~\cite{RS16} a negative and
unexpected result is established: by combining the technique of~\cite{Bezem99,Bezem01}
with the stable model semantics~\cite{GL88}, we get a semantics that is not necessarily
extensional! It remained as an open problem of~\cite{RS16} whether the combination
of the technique of~\cite{Bezem99,Bezem01} with the well-founded approach~\cite{GelderRS91}
leads to an extensional semantics. It is exactly this problem that we undertake to
solve in the present paper.

We demonstrate that the well-founded extension of Bezem's technique {\em fails} to
retain extensionality in the general case. On the positive side, we prove that
for stratified higher-order logic programs, extensionality is indeed achieved. We analyze
the reasons of the failure of extensionality in the general case, and claim that this is
not an inherent shortcoming of Bezem's approach but a more general phenomenon.  In particular
we argue that restricting attention to three-valued logic appears to ``throw away too much
information'' and makes predicates that are expected to have different behaviours,
appear as identical three-valued relations. The main contributions of the present
paper can be summarized as follows:
\begin{itemize}
\item We demonstrate that the well-founded adaptation of Bezem's technique,
      does not in general lead to an extensional model. In particular, we exhibit
      a program with a non-extensional well-founded model. This result, despite
      its negative flavor, indicates that the addition of negation to
      higher-order logic programming is not such a straightforward task as it
      was possibly initially anticipated. Notice that, as it was recently
      demonstrated in~\cite{RS17}, the stable model adaptation of Bezem's
      technique is also non-extensional in general.

\item Despite the above negative result, we prove that the
      well-founded adaptation of Bezem's technique gives an extensional
      two-valued model in the case of {\em stratified} programs.
      This result affirms the importance and the well-behaved nature of stratified
      programs, which was, until now, only known for the first-order case.

\item We study the more general question of the possible existence of an
      {\em alternative} extensional three-valued semantics
      for higher-order logic programs with negation. We indicate that in order to achieve
      such a semantics, one has to make some (arguably) non-standard assumptions
      regarding the behaviour of negation.
\end{itemize}

The rest of the paper is organized as follows. Section~\ref{extensionality} introduces the
basic notions and the advantages of the extensional approach to the semantics of higher-order
logic programming. Section~\ref{intuitive} presents in an intuitive way the main
concepts and results of the paper. Section~\ref{syntax_of_language} introduces
the syntax and Section~\ref{semantics_of_language} the semantics of our source language.
Section~\ref{well_founded} demonstrates that Bezem's approach is not extensional under
the well-founded semantics. In Section~\ref{sec:ext_mon} it is established that stratified
programs have an extensional, two-valued, well-founded semantics. Section~\ref{restrictions}
concludes by discussing the restrictions that any reasonable three-valued semantics would have with respect
to extensionality. The proofs of all results are given in the corresponding appendices.

\section{Extensional Higher-Order Logic Programming}\label{extensionality}
In~\cite{Wa91a} W. W. Wadge suggested that if we appropriately restrict the
syntax of higher-order logic programming, then we can obtain
languages that can be assigned a standard denotational semantics in which
predicates denote sets. In other words, for such syntactically restricted
languages one can apply traditional domain-theoretic notions and tools that
have been used extensively in higher-order functional programming.
The most crucial syntactic restriction imposed by Wadge (and also later
independently by M. Bezem~\cite{Bezem99}), is the following:

\vspace{0.16cm}
\noindent
{\bf The extensionality syntactic restriction:} {\em In the head of every rule in a program,
each argument of predicate type must be a variable, and all such variables must be distinct.}

\begin{example}
The following is a legitimate program that defines the union of two relations {\tt P}, {\tt Q}
(for the moment we use ad-hoc Prolog-like syntax):
\[
\begin{array}{l}
\mbox{\tt union(P,Q)(X):-P(X).}\\
\mbox{\tt union(P,Q)(X):-Q(X).}
\end{array}
\]
However, the following program does not satisfy Wadge's restriction:
\[
\begin{array}{l}
\mbox{\tt q(a).}\\
\mbox{\tt r(q).}
\end{array}
\]
because the predicate constant {\tt q} appears as an argument in the head of a rule. Similarly,
the program:
\[
\begin{array}{l}
\mbox{\tt p(Q,Q):-Q(a).}
\end{array}
\]
is problematic because the predicate variable {\tt Q} is used twice in the
head of the rule.\mathproofbox
\end{example}

The advantages of extensionality were identified by Wadge and Bezem in their respective
papers. First of all, under the extensional approach, program predicates can be understood
declaratively in terms of extensional notions. For example, the program:
\[
\begin{array}{l}
\mbox{\tt map(R,[],[]).}\\
\mbox{\tt map(R,[H1|T1],[H2|T2]):-R(H1,H2), map(R,T1,T2).}
\end{array}
\]
can be understood in a similar way as the well-known {\tt map} function of Haskell.
Moreover, since under the extensional approach predicates denote sets, two predicates
that are true of the same arguments, are considered indistinguishable. So, for example,
if we define two sorting predicates {\tt merge\_sort} and {\tt quick\_sort} that have
the same type, say $\tau$, and that perform the same task (possibly with different efficiency),
it is {\em guaranteed} that any predicate which operates on relations of type $\tau$ will have the same behaviour whether it is given {\tt merge\_sort} or {\tt quick\_sort} as an argument.
 As mentioned in~\cite{Wa91a} ``extensionality means exactly
that predicates are used as {\em black boxes} - and the ``black box'' concept is central
to all kinds of engineering''. It is this property that makes extensional languages so
appealing (and is actually one of the greatest assets of traditional functional programming).

Another important advantage of this declarative approach to higher-order logic programming
is that many techniques and ideas that have been successfully developed in the functional
programming world (such as program transformations, optimizations, techniques for proving
program correctness, and so on), could be transferred to the higher-order logic programming 
domain, opening in this way promising new research directions for logic programming as a whole.

\section{An Intuitive Overview of the Proposed Approach}\label{intuitive}
In this paper we consider the semantic technique for positive higher-order logic programs proposed
in~\cite{Bezem99,Bezem01} and we investigate whether it can be applied in order to provide an extensional
well-founded semantics for higher-order logic programs with negation in clause bodies. In this section
we give an intuitive description of Bezem's idea and we outline how we use it when negation is added
to programs.

Given a positive higher-order logic program, the starting idea behind Bezem's approach
is to take its ``ground instantiation'', in which we replace variables with well-typed
terms that can be created using syntactic entities that appear in the program.
For example, consider the higher-order program below:
\[
\begin{array}{l}
\mbox{\tt q(a).}\\
\mbox{\tt q(b).}\\
\mbox{\tt p(Q):-Q(a).}\\
\mbox{\tt id(R)(X):-R(X).}
\end{array}
\]
In order to obtain the ground instantiation of this program, we consider each clause
and replace each variable of the clause with a ground term that has the same type
as the variable under consideration (the formal definition of this procedure will be given in
Definition~\ref{ground_instantiation_definition}). In this way we obtain the following
infinite program:
\[
\begin{array}{l}
\mbox{\tt q(a).}\\
\mbox{\tt q(b).}\\
\mbox{\tt p(q):-q(a).}\\
\mbox{\tt id(q)(a):-q(a).}\\
\mbox{\tt id(q)(b):-q(b).}\\
\mbox{\tt p(id(q)):-id(q)(a).}\\
\hspace{1.5cm} \cdots
\end{array}
\]
One can now treat the new program as an infinite propositional one (i.e., each ground atom can be seen
as a propositional variable). This implies that we can use the standard least fixed-point construction
of classical logic programming (see for example~\cite{lloyd}) in order to compute the set of atoms
that should be taken as ``true''.  In our example, the least fixed-point will
contain atoms such as {\tt q(a)}, {\tt q(b)}, {\tt p(q)}, {\tt id(q)(a)}, {\tt id(q)(b)}, {\tt p(id(q))},
and so on.

Bezem demonstrated that the least fixed-point semantics of the ground instantiation of every positive
higher-order logic program of the language considered in~\cite{Bezem99,Bezem01}, is {\em extensional}
in a sense that can be explained as follows. In our example, {\tt q} and {\tt id(q)}
are equal since they are both true of exactly the constants {\tt a} and {\tt b}.
Therefore, we expect that (for example) if {\tt p(q)} is true then {\tt p(id(q))}
is also true, because {\tt q} and {\tt id(q)} should be considered as indistinguishable.
This property of ``indistinguishability'' is formally defined in~\cite{Bezem99,Bezem01} and it is
demonstrated that it holds in the least fixed-point of the immediate consequence operator
of the ground instantiation of every program that abides to the simple extensionality
syntactic restriction given in the previous section (and formally described by Definition \ref{def:programs}
later in the paper).


The key idea behind extending Bezem's semantics in order to apply to higher-order logic programs
with negation, is straightforward to state: given such a program, we first take its ground instantiation.
The resulting program is a (possibly infinite) propositional program with negation
and therefore we can compute its semantics in any standard way that exists for obtaining
the meaning of such programs. For example, one could use the well-founded semantics~\cite{GelderRS91},
the stable model semantics~\cite{GL88}, or the infinite-valued semantics~\cite{RondogiannisW05},
and then proceed to examine whether the well-founded model (respectively, each stable model, or the
minimum infinite-valued model) is extensional in the sense of~\cite{Bezem99,Bezem01} (informally
described above).


An open problem posed in~\cite{RS17} was
whether Bezem's technique, under the well-founded semantics, always leads to
an extensional well-founded model. As we are going to see in the subsequent sections,
this is not the case. In particular, we exhibit a program containing three
predicates {\tt s}, {\tt p} and {\tt q}, such that {\tt p} and {\tt q} are extensionally
equal under the well-founded semantics, but {\tt s(p)} and {\tt s(q)}
have a different truth value. On the positive side,
we prove that every stratified higher-order logic program with negation has an extensional
well-founded model. In this sense, we identify a broad class of programs
that are well-behaved in terms of extensionality.


\section{The Syntax of \texorpdfstring{$\mathcal{H}$}{H}}\label{syntax_of_language}
In this section we define the syntax of the language $\mathcal{H}$ that we use throughout the paper.
$\mathcal{H}$ is based on a simple type system with two base types: $o$, the boolean domain, and $\iota$, the domain of data objects.
The composite types are partitioned into three classes: functional (assigned to function symbols), predicate (assigned to predicate symbols)
and argument (assigned to parameters of predicates).
\begin{definition}
A type can either be \emph{functional}, \emph{predicate}, or \emph{argument}, denoted by $\sigma$, $\pi$
and $\rho$ respectively and defined as:
\begin{align*}
\sigma & := \iota \mid (\iota \rightarrow \sigma) \\
\pi & := o \mid (\rho \rightarrow \pi) \\
\rho & := \iota \mid \pi
\end{align*}
\end{definition}

We will use $\tau$ to denote an arbitrary type (either functional, predicate, or argument). As usual, the binary operator $\rightarrow$ is right-associative. A functional type that is different than $\iota$ will often be written in the form $\iota^n \rightarrow \iota$, $n\geq 1$. Moreover, it can be easily seen that every predicate type $\pi$ can be written in the form $\rho_1 \rightarrow \cdots \rightarrow \rho_n \rightarrow o$, $n\geq 0$ (for $n=0$ we assume that $\pi=o$).  We proceed by defining the syntax of $\mathcal{H}$:
\begin{definition}
The \emph{alphabet} of $\mathcal{H}$ consists of the following: {\em predicate variables}
of every predicate type $\pi$ (denoted by capital letters such as $\mathsf{Q,R,}$ $\mathsf{S,\ldots}$);
{\em individual variables} of type $\iota$ (denoted by capital letters such as $\mathsf{X,Y,Z,\ldots}$);
{\em predicate constants} of every predicate type $\pi$ (denoted by lowercase letters such as $\mathsf{p,q,r,\ldots}$);
{\em individual constants} of type $\iota$ (denoted by lowercase letters such as $\mathsf{a,b,c,\ldots}$);
{\em function symbols} of every functional type $\sigma \neq \iota$ (denoted by lowercase letters such as $\mathsf{f,g,h,\ldots}$);
the {\em inverse implication} constant $\leftarrow$; the {\em negation} constant $\pnot$; the comma;
the left and right parentheses; and the {\em equality} constant $\approx$ for comparing terms of type $\iota$.
\end{definition}


Arbitrary variables will be usually denoted by $\mathsf{V}$ and its subscripted versions.
\begin{definition}
The set of {\em terms} of $\mathcal{H}$ is defined as follows: every predicate variable (respectively,
predicate constant) of type $\pi$ is a term of type $\pi$; every individual variable (respectively,
individual constant) of type $\iota$ is a term of type $\iota$; if $\mathsf{f}$ is an $n$-ary function
symbol and $\mathsf{E}_1, \ldots, \mathsf{E}_n$ are terms of type $\iota$ then
$(\mathsf{f}\ \mathsf{E}_1\cdots\mathsf{E}_n)$ is a term of type $\iota$; if $\mathsf{E}_1$ is a term
of type $\rho \rightarrow \pi$ and $\mathsf{E}_2$ a term of type $\rho$ then $(\mathsf{E}_1\ \mathsf{E}_2)$
is a term of type $\pi$.
\end{definition}

%
\begin{definition}
\label{def:expressions}
The set of {\em expressions} of $\mathcal{H}$ is defined as follows: a term of type $\rho$ is an expression of type $\rho$;
if $\mathsf{E}$ is a term of type $o$ then $(\pnot \mathsf{E})$ is an expression of type $o$;
if $\mathsf{E}_1$ and $\mathsf{E}_2$ are terms of type $\iota$, then $(\mathsf{E}_1\approx \mathsf{E}_2)$ is an expression of type $o$.
\end{definition}
%

%

We will omit parentheses
when no confusion arises. To denote that an expression $\mathsf{E}$ has type $\rho$ we will often write $\mathsf{E}:\rho$. We will write $vars(\mathsf{E})$ to denote the set of all the variables in $\mathsf{E}$.
Expressions (respectively, terms) that have no variables will be referred to as {\em ground expressions} (respectively, {\em ground terms}).
Terms of type $o$ will be referred to as {\em atoms}. Expressions of type $o$ that do not contain negation, i.e. expressions of the form $(\mathsf{E}_1\approx \mathsf{E}_2)$ or atoms, will be called {\em positive literals}, while expressions of the form $(\pnot \mathsf{E})$ will be called {\em negative literals}. A {\em literal} is either a positive literal or a negative literal.
\begin{definition}\label{def:programs}
A {\em clause} of $\mathcal{H}$ is a formula
$\mathsf{p}\ \mathsf{V}_1 \cdots \mathsf{V}_n \leftarrow \mathsf{L}_1, \ldots, \mathsf{L}_m$,
where $\mathsf{p}$ is a predicate constant of type $\rho_1 \rightarrow \cdots \rightarrow\rho_n \rightarrow o$, $\mathsf{V}_1,\ldots,\mathsf{V}_n$ are distinct variables of types $\rho_1,\ldots,\rho_n$ respectively and $\mathsf{L}_1,\ldots,\mathsf{L}_m$ are literals. The term $\mathsf{p}\ \mathsf{V}_1 \cdots \mathsf{V}_n$ is called the {\em head} of the clause, the variables $\mathsf{V}_1, \ldots, \mathsf{V}_n$ are the {\em formal parameters} of the
clause
and the conjunction $\mathsf{L}_1,\ldots, \mathsf{L}_m$ is its {\em body}.
A {\em program} $\mathsf{P}$ of $\mathcal{H}$ is a finite set of clauses.
\end{definition}
\begin{example}
The program below defines the {\tt subset} relation over unary predicates:
\[
\begin{array}{l}
\mbox{\tt subset S1 S2 $\leftarrow$ $\pnot$(nonsubset S1 S2)}\\
\mbox{\tt nonsubset S1 S2 $\leftarrow$ (S1 X), $\pnot$(S2 X)}
\end{array}
\]
Given unary predicates {\tt p} and {\tt q}, {\tt subset p q} is true iff {\tt p} is a subset of {\tt q}.\mathproofbox
\end{example}
\begin{example}
\label{exm:prefernce}
For a more ``real-life'' higher-order logic program with negation, assume that we have a unary
predicate {\tt movie M} and a binary predicate {\tt ranking M R} which returns the ranking
{\tt R} of a given movie {\tt M}. Consider also the following first-order predicate that
defines a preference over movies based on their ranking:
\[
\begin{array}{l}
\mbox{\tt prefer M1 M2 $\leftarrow$ movie M1, movie M2, ranking M1 R1, ranking M2 R2, R1>R2.}
\end{array}
\]
The following higher-order predicate {\tt winnow} (see for example~\cite{Chomicki03}) can be used
to select all the ``best'' tuples {\tt T} out of a given relation {\tt R} based on a preference
relation {\tt P}:
\[
\begin{array}{l}
\mbox{\tt winnow P R T $\leftarrow$ R T, $\pnot$(bypassed P R T).}\\
\mbox{\tt bypassed P R T $\leftarrow$ R T1, P T1 T.}
\end{array}
\]
Intuitively, {\tt winnow} returns all the tuples {\tt T} of the relation {\tt R} such that
there does not exist any tuple {\tt T1} in the relation {\tt R} that is better from {\tt T}
with respect to the preference relation {\tt P}. For example, if we ask the query
{\tt ?- winnow prefer movie T.} we expect as answers all those movies that have the highest
possible ranking. Notice that since {\tt winnow} is a higher-order predicate, it can be invoked
with different arguments; for example, it can be used to select out of a {\tt book} relation,
all those books that have the lowest possible price, or out of a {\tt flight} relation all those
flights that go to London, and so on.\mathproofbox
\end{example}


The ground instantiation of a program is described by the following definitions:
\begin{definition}
A {\em substitution} $\theta$ is a finite set of the form $\{ \mathsf{V}_1/\mathsf{E}_1, \ldots, \mathsf{V}_n/\mathsf{E}_n\}$
where the $\mathsf{V}_i$'s are different variables and each $\mathsf{E}_i$
is a term having the same type as $\mathsf{V}_i$. We write
$dom(\theta)$ to denote the domain $\{ \mathsf{V}_1, \ldots, \mathsf{V}_n\}$ of $\theta$. If
all the terms $\mathsf{E}_1, \ldots, \mathsf{E}_n$ are ground, $\theta$ is called a {\em ground substitution}.
\end{definition}
\begin{definition}
\label{def:ground_inst}
Let $\theta$ be a substitution and $\mathsf{E}$ be an expression. Then, $\mathsf{E}\theta$
is an expression obtained from $\mathsf{E}$ as follows:
\begin{itemize}
  \item $\mathsf{E}\theta = \mathsf{E}$ if $\mathsf{E}$ is a predicate constant or individual constant;
  \item $\mathsf{V}\theta = \theta(\mathsf{V})$ if $\mathsf{V} \in dom(\theta)$; otherwise, $\mathsf{V}\theta = \mathsf{V}$;
  \item $(\mathsf{f}\ \mathsf{E}_1\cdots\mathsf{E}_n)\theta = (\mathsf{f}\ \mathsf{E}_1\theta\cdots\mathsf{E}_n\theta)$;
  \item $(\mathsf{E}_1\ \mathsf{E}_2)\theta = (\mathsf{E}_1\theta\ \mathsf{E}_2\theta)$;
  \item $(\pnot \mathsf{E})\theta= (\pnot \mathsf{E}\theta)$;
  \item $(\mathsf{E}_1\approx \mathsf{E}_2)\theta = (\mathsf{E}_1\theta\approx \mathsf{E}_2\theta)$.
\end{itemize}
If $\theta$ is a ground substitution such that $vars(\mathsf{E}) \subseteq dom(\theta)$, then
the ground expression $\mathsf{E}\theta$ is called a {\em ground instance} of $\mathsf{E}$.
\end{definition}
\begin{definition}\label{ground_instantiation_definition}
Let $\mathsf{P}$ be a program. A {\em ground instance of a clause}
$\mathsf{p}\ \mathsf{V}_1 \cdots \mathsf{V}_n \leftarrow \mathsf{L}_1,\ldots,\mathsf{L}_m$
of $\mathsf{P}$ is a formula  $(\mathsf{p}\ \mathsf{V}_1 \cdots \mathsf{V}_n)\theta \leftarrow \mathsf{L}_1\theta,\ldots,\mathsf{L}_m\theta$, where $\theta$ is a ground substitution whose domain
is the set of all variables that appear in the clause, such that for every $\mathsf{V} \in dom(\theta)$ with $\mathsf{V}:\rho$,
$\theta(\mathsf{V})$ is a ground term of type $\rho$ that has been formed with predicate constants, function symbols,
and individual constants that appear in $\mathsf{P}$. The {\em ground instantiation of a program} $\mathsf{P}$,
denoted by $\mathsf{Gr(P)}$, is the (possibly infinite) set that contains all the ground instances of the
clauses of $\mathsf{P}$.
\end{definition}
%

\section{The Semantics of \texorpdfstring{$\mathcal{H}$}{H}}\label{semantics_of_language}
In~\cite{Bezem99,Bezem01} M. Bezem developed a semantics for higher-order logic programs which
generalizes the familiar Herbrand-model semantics of classical (first-order) logic programs.
In this section we extend Bezem's semantics to the case of higher-order logic programs with negation.

In order to interpret the programs of  $\mathcal{H}$, we need to specify the semantic domains in which the expressions of each type $\tau$ are assigned their meanings. The following definition is a slightly modified version of the corresponding definition of~\cite{Bezem99,Bezem01} and it implies that the expressions of predicate types should be understood as representing functions. We use $[S_1 \rightarrow S_2]$ to denote the set of (possibly partial) functions from a set $S_1$ to a set $S_2$. The possibility to have a partial function arises due to a technicality
which is explained in the remark just above Definition~\ref{HInterp}.
\begin{definition}
\label{def:type_structure}
A \emph{functional type structure} $\mathcal{S}$ for $\mathcal{H}$ consists of two non-empty sets $D$ and $A$ together with an assignment
$\mo{\tau}$ to each type $\tau$ of $\mathcal{H}$, so that the following are satisfied:

\begin{itemize}
\item $\mo{\iota} = D$;
\item $\mo{\iota^n \rightarrow \iota} = D^n \rightarrow D$;
\item $\mo{o} = A$;
\item $\mo{\rho \rightarrow \pi} \subseteq [\mo{\rho} \rightarrow \mo{\pi}]$.
\end{itemize}
\end{definition}

Given a functional type structure $\mathcal{S}$, any function $v:\mo{o}\rightarrow \{\mathit{false}, \unk, \mathit{true}\}$ will be called a \emph{three-valued valuation function} (or simply \emph{valuation function}) for $\mathcal{S}$. We will use the term \emph{two-valued valuation functions} to distinguish the subset of valuation functions which do not assign the value $\unk$ to any element of $\mo{o}$, i.e. the functions $v:\mo{o}\rightarrow \{\mathit{false}, \mathit{true}\}$.

It is customary in the study of the semantics of logic programming languages to restrict attention to
{\em Herbrand interpretations}. Given a program $\mathsf{P}$, a Hebrand interpretation is one that
has as its underlying universe the so-called {\em Herbrand universe} of $\mathsf{P}$:
\begin{definition}\label{HUniverse}
For a program $\mathsf{P}$, we define the \emph{Herbrand universe} for every argument type $\rho$, denoted by
$U_{\mathsf{P},\rho}$ to be the set of all ground terms of type $\rho$ that can be formed out of the individual constants, function symbols, and predicate constants in the program. Moreover, we define $U^+_{\mathsf{P},o}$ to be the set of all ground expressions of type $o$, that can be formed out of the above symbols, i.e. the set $U^+_{\mathsf{P},o}=U_{\mathsf{P},o} \cup \{(\mathsf{E}_1\approx \mathsf{E}_2)\mid\mathsf{E}_1, \mathsf{E}_2\in U_{\mathsf{P},\iota}\} \cup \{(\pnot \mathsf{E})\mid\mathsf{E}\in U_{\mathsf{P},o}\}$.
\end{definition}

Following~\cite{Bezem99,Bezem01}, we take $D$ and $A$ in Definition \ref{def:type_structure} to be equal to $U_{\mathsf{P},\iota}$ and $U^+_{\mathsf{P},o}$ respectively. Then, for each predicate type $\rho \rightarrow \pi$, each element of $U_{\mathsf{P}, \rho \rightarrow \pi}$ can be perceived as a function mapping elements of $\mo{\rho}$ to elements of $\mo{\pi}$, through syntactic application mapping. That is, $\mathsf{E}\in U_{\mathsf{P},\rho \rightarrow \pi}$ can be viewed as the function mapping each term $\mathsf{E}' \in  U_{\mathsf{P},\rho}$ to the term $(\mathsf{E} \, \mathsf{E}') \in  U_{\mathsf{P},\pi}$. Similarly, every $n$-ary function symbol $\mathsf{f}$ appearing in $\mathsf{P}$ can be viewed as the function mapping each element $(\mathsf{E}_1, \ldots,\mathsf{E}_n) \in U_{\mathsf{P},\iota}^n$ to the term $(\mathsf{f} \, \mathsf{E}_1\,\cdots\,\mathsf{E}_n) \in U_{\mathsf{P},\iota}$.

\vspace{0.2cm}
\noindent
{\bf Remark:} There is a small technicality here which we need to clarify. In the case where $\rho=o$,
$\mathsf{E}\in U_{\mathsf{P},o \rightarrow \pi}$ is a partial function because it maps elements of
$U_{\mathsf{P},o}$ (and not of $U^+_{\mathsf{P},o}$) to elements of $U_{\mathsf{P},\pi}$; this is due
to the fact that our syntax does not allow an expression of type $o \rightarrow \pi$ to take as argument
an expression of the form $(\mathsf{E}_1\approx \mathsf{E}_2)$ nor of the form $(\pnot \mathsf{E})$.
In all other cases (i.e., when $\rho\neq o$), $\mathsf{E}$ represents a total function.
\begin{definition}\label{HInterp}
A \emph{(three-valued) Herbrand interpretation} $I$ of a program $\mathsf{P}$ consists of:
\begin{enumerate}
\item the functional type structure $\mathcal{S}_\mathsf{P}$, such that $D= U_{\mathsf{P},\iota}$, $A=U^+_{\mathsf{P},o}$ and $\mo{\rho \rightarrow \pi} =U_{\mathsf{P},\rho \rightarrow \pi}$ for every predicate type $\rho \rightarrow \pi$, called the Herbrand type structure of $\mathsf{P}$;
\item the assignment to each individual constant $\mathsf{c}$ in $\mathsf{P}$, of the element $I(\mathsf{c}) = \mathsf{c}$; to each predicate constant $\mathsf{p}$ in $\mathsf{P}$, of the element $I(\mathsf{p}) =\mathsf{p}$; to each function symbol $\mathsf{f}$ in $\mathsf{P}$, of the element $I(\mathsf{f}) = \mathsf{f}$;
\item a valuation function $\vwrt{\cdot}{I}$ for $\mathcal{S}_\mathsf{P}$,  assigning to each element of $U^+_{\mathsf{P},o}$ an element in  $\{\mathit{false}, \unk, \mathit{true}\}$, while satisfying the following:
	\begin{itemize}
	\item for all $\mathsf{E}_1,\mathsf{E}_2\in U_{\mathsf{P},\iota}$, $\vwrt{(\mathsf{E}_1\approx\mathsf{E}_2)}{I}= \begin{cases}
\mathit{false}, & \mbox{if } \mathsf{E}_1\neq \mathsf{E}_2\\
\mathit{true}, & \mbox{if } \mathsf{E}_1= \mathsf{E}_2
\end{cases}$;
	\item for all $\mathsf{E}\in U_{\mathsf{P},o}$, $\vwrt{(\mathsf{\pnot E)}}{I} = \begin{cases}
\mathit{false}, & \mbox{if } \vwrt{\mathsf{E}}{I}=\mathit{true}\\
\unk, & \mbox{if } \vwrt{\mathsf{E}}{I}=\unk\\
\mathit{true}, & \mbox{if } \vwrt{\mathsf{E}}{I}=\mathit{false}
\end{cases}$.
	\end{itemize}
\end{enumerate}
\end{definition}
We call $\vwrt{\cdot}{I}$ the {\em valuation function of $I$} and omit the reference to $\mathcal{S}_\mathsf{P}$, since the latter is common to all Herbrand interpretations of a program. In fact, individual Herbrand interpretations are only set apart by their valuation functions. If the valuation function $\vwrt{\cdot}{I}$ is two-valued, then $I$ will also be called a \emph{two-valued Herbrand interpretation}.
\begin{definition}
  A \emph{Herbrand state} (or simply \emph{state}) $s$ of a program $\mathsf{P}$ is a function that assigns to each variable $\mathsf{V}$ of type $\rho$ an element of $U_{\mathsf{P},\rho}$.
\end{definition}
Given a Herbrand interpretation $I$ and state $s$, we can define the semantics of expressions with respect to $I$ and $s$.
\begin{definition}
Let $\mathsf{P}$ be a program. Also, let $I$ be a Herbrand interpretation and $s$ a
Herbrand state of $\mathsf{P}$. Then the semantics of expressions with respect to $I$ and $s$ is defined as follows:
\begin{itemize}
\item $\mwrs{\mathsf{c}}{I}{s} = I(\mathsf{c})=\mathsf{c}$, for every individual constant $\mathsf{c}$;
\item $\mwrs{\mathsf{p}}{I}{s} =I(\mathsf{p})=\mathsf{p}$, for every predicate constant $\mathsf{p}$;
\item $\mwrs{\mathsf{V}}{I}{s} = s(\mathsf{V})$, for every variable $\mathsf{V}$;
\item $\mwrs{(\mathsf{f}\, \mathsf{E}_1\, \cdots\,\mathsf{E}_n)}{I}{s}
       =(I(\mathsf{f})\, \mwrs{\mathsf{E_1}}{I}{s}\, \cdots\,\mwrs{\mathsf{E_n}}{I}{s})
     = (\mathsf{f}\, \mwrs{\mathsf{E}_1}{I}{s}\, \cdots\,\mwrs{\mathsf{E}_n}{I}{s})$, for every $n$-ary function symbol $\mathsf{f}$;
\item $\mwrs{(\mathsf{E}_1\,\mathsf{E}_2)}{I}{s}= (\mwrs{\mathsf{E}_1}{I}{s}\,\mwrs{\mathsf{E}_2}{I}{s})$;
\item $\mwrs{(\mathsf{E}_1\approx \mathsf{E}_2)}{I}{s}= (\mwrs{\mathsf{E}_1}{I}{s}\approx\mwrs{\mathsf{E}_2}{I}{s})$;
\item $\mwrs{\mathsf{(\pnot E)}}{I}{s}= (\pnot \mwrs{\mathsf{E}}{I}{s}).$
\end{itemize}
\end{definition}
It is easy to see that the semantic function $\mo{\cdot}$ is well defined, in the sense that, for every Herbrand state $s$ and every expression $\mathsf{E}$ of every argument type $\rho$, we have $\mwrs{\mathsf{E}}{I}{s}\in\mo{\rho}$. Note that this makes $\mwrs{\mathsf{E}}{I}{s}$ a ground expression of the language. Also note, that if $\mathsf{E}$ is a ground expression then $\mwrs{\mathsf{E}}{I}{s} = \mathsf{E}$;
therefore, if $\mathsf{E}$ is a ground literal, we can write $\vwrt{\mathsf{E}}{I}$ instead of $\vwrt{\mwrs{\mathsf{E}}{I}{s}}{I}$.
Stretching this abuse of notation a little further, we can extend a valuation function to assign truth values to ground conjunctions of literals; this allows us to define the concept of Herbrand models for our higher-order programs in the same way as in classical logic programming.
\begin{definition}
Let $\mathsf{P}$ be a program and $I$ be a Herbrand interpretation of $\mathsf{P}$. We define $\vwrt{\mathsf{L}_1, \ldots, \mathsf{L}_n}{I} = min\{\vwrt{\mathsf{L}_1}{I}, \ldots, \vwrt{\mathsf{L}_n}{I}\}$
for all $\mathsf{L}_1, \ldots, \mathsf{L}_n\in U^+_{\mathsf{P},o}$. Moreover, we say $I$ is a \emph{model} of $\mathsf{P}$ if $\vwrt{\mwrs{\mathsf{A}}{I}{s}}{I}\geq\vwrt{\mwrs{\mathsf{L}_1}{I}{s},\ldots, \mwrs{\mathsf{L}_m}{I}{s}}{I}$ holds for every clause $\mathsf{A} \leftarrow \mathsf{L}_1,\ldots, \mathsf{L}_m$ and every Herbrand state $s$ of $\mathsf{P}$.
\end{definition}

Bezem's semantics is based on the observation that, given a positive higher-order program, we can use the minimum model of its ground instantiation as a (two-valued) valuation function defining a Herbrand interpretation for the initial program itself.
We follow the same idea but now for programs with negation: we can use as the valuation function of a given $\mathcal{H}$ program,
the Herbrand model defined by any semantic approach that applies to its ground instantiation. Actually, we demonstrate (see Theorem~\ref{thr:min_mod}
below) that any interpretation $I$ of the higher-order program $\mathsf{P}$ will be a minimal model of $\mathsf{P}$,
if its chosen valuation function is a minimal model of $\mathsf{Gr(P)}$.

We consider two different notions of minimality, based on the \emph{truth} ordering $\tleq$ and the \emph{Fitting} ordering $\fleq$ of truth values, respectively. Recall that $\tleq$ is the partial order defined by $\mathit{false}\tleq \unk \tleq \mathit{true}$, while $\fleq$ is the partial order defined by $\unk\fleq \mathit{false}$ and $\unk\fleq \mathit{true}$.
\begin{definition}
If $I$ and $J$ are two Herbrand interpretations of a higher-order program $\mathsf{P}$, we say $I\tleq J$ (respectively, $I\fleq J$) if, for all atoms $\mathsf{A}$ in $U_{\mathsf{P},o}$ we have $v_I(\mathsf{A})\tleq v_J(\mathsf{A})$ (resp., $v_I(\mathsf{A})\fleq v_J(\mathsf{A})$). If $M$ is a model of $\mathsf{P}$ then we say it is $\tleq$-minimal (resp., $\fleq$-minimal) if there does not exist a different model $N$ of $\mathsf{P}$, such that $N\tleq M$ (resp., $N\fleq M$).
\end{definition}
%
%
%
\begin{theorem}
\label{thr:min_mod}
Let $\mathsf{P}$ be a program and let $\mathsf{Gr(P)}$ be its ground instantiation. Also let $M$ be a partial interpretation\repnote{\ }of $\mathsf{Gr(P)}$ and let $\mathcal{M}$ be the Herbrand interpretation of $\mathsf{P}$, such that $\vwrt{\mathsf{A}}{\mathcal{M}} = M(\mathsf{A})$ for every $\mathsf{A}\in U_{\mathsf{P},o}$.
Then, $\mathcal{M}$ is a Herbrand model of $\mathsf{P}$ if and only if $M$ is a model of $\mathsf{Gr(P)}$. Moreover, $\mathcal{M}$ is $\tleq$-minimal (respectively, $\fleq$-minimal) if and only if $M$ is $\tleq$-minimal (respectively, $\fleq$-minimal).
\end{theorem}

As an application of the above developments we define two special Herbrand interpretations of higher-order programs, employing the well-known perfect model~\cite{AptBW88,Gelder89} and well-founded model~\cite{GelderRS91} of the ground instantiation of a program, as valuation functions.
\begin{definition}
Let $\mathsf{P}$ be a program and let $\mathsf{Gr(P)}$ be the ground instantiation of $\mathsf{P}$. Also, let $\pgp$ be the perfect model\repnote{\ }(if this exists) and $\wgp$ be the well-founded model\repnote{\ }of $\mathsf{Gr(P)}$. We define $\pmp$ to be the two-valued Herbrand interpretation of $\mathsf{P}$ such that $\vwrt{\mathsf{A}}{\pmp} = \pgp(\mathsf{A})$ for every $\mathsf{A}\in U_{\mathsf{P},o}$. Similarly, we define $\wfmp$ to be the three-valued Herbrand interpretation of $\mathsf{P}$ such that $\vwrt{\mathsf{A}}{\wfmp} = \wgp(\mathsf{A})$ for every $\mathsf{A}\in U_{\mathsf{P},o}$.
\end{definition}

Clearly, by Theorem~\ref{thr:min_mod}, $\pmp$ (if it exists) is a two-valued minimal model and $\wfmp$ is a three-valued minimal model of $\mathsf{P}$. In the following sections we investigate their suitability for providing extensional semantics for $\mathcal{H}$ programs.
In particular, we examine if each of them enjoys the extensionality property, formally defined~\cite{Bezem99,Bezem01} through relations $\exeq[v,\tau]$
over the set of expressions of a given type $\tau$ and under a given valuation function $v$. These relations
intuitively express extensional equality of type $\tau$, in the sense discussed in Section~\ref{intuitive}. For the purposes of this paper only extensional equality of argument types will be needed, for which the formal definition is as follows:
\begin{definition}\label{def:extensional_eq}
Let $\mathcal{S}$ be a functional type structure and $v$ be a valuation function for $\mathcal{S}$. For every argument type $\rho$ we define the relations $\exeq[v,\rho]$ on $\mo{\rho}$ as follows: Let $d, d'\in \mo{\rho}$; then $d \exeq[v,\rho] d'$ if and only if
\begin{enumerate}
\item $\rho = \iota$ and $d=d'$, or
\item $\rho = o$ and $v(d) = v(d')$, or
\item $\rho = \rho'\rightarrow\pi$ and $d\,e \exeq[v,\pi] d'\,e'$ for all $e, e'\in\mo{\rho'}$, such that $e\exeq[v,\rho'] e'$ and $d\,e, d'\,e'$ are both defined.
\end{enumerate}
\end{definition}
One can easily verify that, for all $d, d'\in \mo{\rho_1\rightarrow\cdots\rightarrow\rho_n\rightarrow o}$, $e_1, e_1' \in \mo{\rho_1}$, \ldots, $e_n, e_n'\in \mo{\rho_n}$, if $d\exeq[v,\rho_1\rightarrow\cdots\rightarrow\rho_n\rightarrow o] d'$, $e_1\exeq[v,\rho_1] e_1'$, \ldots, $e_n\exeq[v,\rho_n] e_n'$ and $d\ e_1\ \cdots\ e_n, d'\ e_1'\ \cdots\ e_n'$ are both defined, then $v(d\ e_1\ \cdots\ e_n) = v(d'\ e_1'\ \cdots\ e_n')$.

Generally, it is not guaranteed that such relations will be equivalence relations; rather they are partial equivalences (they are shown in \cite{Bezem99} to be symmetric and transitive). Whether they are moreover reflexive, depends on the specific valuation function.

The above discussion leads to the notion of {\em extensional interpretation}:
\begin{definition}\label{def:extensionality}
Let $\mathsf{P}$ be a program and let $I$ be a Herbrand interpretation of $\mathsf{P}$ with valuation function $v_I$. We say $I$ is \emph{extensional} if for all argument types $\rho$  the relations $\exeq[v_I, \rho]$ are reflexive, i.e. for all $\mathsf{E} \in \mo{\rho}$, it holds that $\mathsf{E} \exeq[v_I, \rho] \mathsf{E}$.
\end{definition}
The above notion will be extensively used in the following two sections.

\section{Non-Extensionality of the Well-Founded Model}\label{well_founded}
In this section we demonstrate that the adaptation of Bezem's technique under the well-founded semantics
does not in general preserve extensionality. In particular, we exhibit below a program that has
a non-extensional well-founded model.
\begin{example}\label{exm:not_extens}
Consider the higher-order program $\mathsf{P}$:
\[
\begin{array}{l}
\mbox{\tt s Q $\leftarrow$ Q (s Q)}\\
\mbox{\tt p R $\leftarrow$ R}\\
\mbox{\tt q R $\leftarrow$ $\pnot$(w R)}\\
\mbox{\tt w R $\leftarrow$ $\pnot$R}
\end{array}
\]
where the predicate variable {\tt Q} is of type $o\rightarrow o$ and the predicate variable {\tt R}
is of type $o$. Before stating formally the non-extensionality result, certain explanations
at an intuitive level are in order. Consider first the predicate {\tt p} of type $o\rightarrow o$.
One can view {\tt p} as representing the identity relation on truth values, i.e., as the relation
$\{(v,v)\mid v\in\{\mathit{false},0,\mathit{true}\}\}$. It is not hard to see that the predicate
{\tt q} of type $o\rightarrow o$, represents exactly the same relation. However, the definition of
{\tt q} involves two applications of negation, while {\tt p} is defined directly (without the use
of negation).

Consider now the predicate {\tt s} of type $(o\rightarrow o) \rightarrow o$ which can take as a
parameter either {\tt p} or {\tt q}. When {\tt s} takes {\tt p} as a parameter, we get the following
two clauses (by substituting {\tt p} for {\tt Q} and {\tt (s p)} for {\tt R} in the above program):
\[
\begin{array}{l}
\mbox{\tt s p $\leftarrow$ p (s p)}\\
\mbox{\tt p (s p) $\leftarrow$ (s p)}
\end{array}
\]
A recursive definition of this form assigns to {\tt (s p)}, under the well-founded semantics,
the value $\mathit{false}$. Consider on the other hand the case where {\tt s} takes {\tt q}
as a parameter. Then, by doing analogous substitutions, we get the following three clauses:
\[
\begin{array}{l}
\mbox{\tt s q $\leftarrow$ q (s q)}\\
\mbox{\tt q (s q) $\leftarrow$ $\pnot$(w (s q))}\\
\mbox{\tt w (s q) $\leftarrow$ $\pnot$(s q)}
\end{array}
\]
Under the well-founded semantics, {\tt (s q)} is assigned the value $0$.
In other words, despite the fact that {\tt p} and {\tt q} are extensionally
equal (see also below), {\tt (s p)} and {\tt (s q)} have
different truth values. In conclusion, the adaptation of the well-founded
semantics under Bezem's technique does not lead to an extensional model
in all cases.\mathproofbox
\end{example}
Of course, the above discussion is based on intuitive arguments, but it
is not hard to formalize it. The main difficulty lies in establishing that
{\tt p} and {\tt q} are extensionally equal because the above program
has an infinite ground instantiation $\mathsf{Gr(P)}$ (see ~\ref{app:counterexample}).
The following lemma, whose detailed proof is given in~\ref{app:counterexample}, suggests that $\wfmp$, i.e. the Herbrand
interpretation of our example program $\mathsf{P}$ defined by using the well-founded model $\wgp$ of
$\mathsf{Gr(P)}$ as the valuation function, is not extensional.
\begin{lemma}
\label{lm:exm_is_not_extens}
The Herbrand interpretation $\wfmp$ of the program of Example~\ref{exm:not_extens} is not extensional.
\end{lemma}
The consequences that the above lemma has for the investigation of alternative extensional
three-valued semantics for higher-order logic programs with negation, will be discussed in Section~\ref{restrictions}.

A natural question that arises is whether there exists a broad and useful class of programs that
are extensional under the well-founded semantics. The next section answers exactly this question.

\section{Extensionality of Stratified Programs}\label{sec:ext_mon}
In this section we present the notion of {\em stratified} higher-order programs, originally introduced in~\cite{RS16}, and argue that the well-founded model of such a program enjoys the extensionality property defined in Section~\ref{semantics_of_language}.
In the following definition, a predicate type $\pi$ is understood to be \emph{greater than} a second predicate type $\pi'$, if $\pi$ is of the form $\rho_1\rightarrow\cdots\rightarrow\rho_n\rightarrow\pi'$, where $n\geq1$.
\begin{definition}\label{stratified}
A program $\mathsf{P}$ is called {\em stratified} if and only if it is possible to decompose the
set of all predicate constants that appear in $\mathsf{P}$ into a finite number $r$ of disjoint sets (called {\em strata})
$S_1,S_2,\ldots,S_r$, such that for every clause
$\mathsf{H} \leftarrow \mathsf{A}_1,\ldots,\mathsf{A}_m,\pnot \mathsf{B}_1,\ldots,\pnot \mathsf{B}_n$
in $\mathsf{P}$, where the predicate constant of $\mathsf{H}$ is $\mathsf{p}$, we have:
\begin{enumerate}
\item for every $i\leq m$, if $\mathsf{A}_i$ is a term that starts with a predicate constant $\mathsf{q}$, then
      $\textit{stratum}(\mathsf{q}) \leq \textit{stratum}(\mathsf{p})$;

\item for every $i\leq m$, if $\mathsf{A}_i$ is a term that starts with a predicate variable $\mathsf{Q}$, then
      for all predicate constants $\mathsf{q}$ that appear in $\mathsf{P}$ such that the type of
      $\mathsf{q}$ is greater than or equal to the type of $\mathsf{Q}$, it holds
      $\textit{stratum}(\mathsf{q}) \leq \textit{stratum}(\mathsf{p})$;

\item for every $i\leq n$, if $\mathsf{B}_i$ starts with a predicate constant $\mathsf{q}$, then
      $\textit{stratum}(\mathsf{q}) < \textit{stratum}(\mathsf{p})$;

\item for every $i\leq n$, if $\mathsf{B}_i$ starts with a predicate variable $\mathsf{Q}$, then
      for all predicate constants $\mathsf{q}$ that appear in $\mathsf{P}$ such that the type of
      $\mathsf{q}$ is greater than or equal to the type of $\mathsf{Q}$, it holds
      $\textit{stratum}(\mathsf{q}) < \textit{stratum}(\mathsf{p})$;
\end{enumerate}
where  $\textit{stratum}(\mathsf{r}) = i$ if the predicate constant
$\mathsf{r}$ belongs to $S_i$.
\end{definition}
One may easily see that the stratification for classical logic programs~\cite{AptBW88,Gelder89} is a special case of the above definition.
\begin{example}
It is straightforward to see that the program:
\[
\begin{array}{l}
\mbox{\tt p Q $\leftarrow$ $\pnot$(Q a)}\\
\mbox{\tt q X $\leftarrow$ (X$\approx$a)}
\end{array}
\]
is stratified. However, it can easily be checked that the program:
\[
\begin{array}{l}
\mbox{\tt p Q $\leftarrow$ $\pnot$(Q a)}\\
\mbox{\tt q X Y $\leftarrow$ (X$\approx$a),(Y$\approx$a),p (q a)}
\end{array}
\]
%
%
is not stratified because if the term {\tt (q a)} is substituted for {\tt Q} we get
a circularity through negation. Notice that the type of {\tt q} is $\iota\rightarrow \iota \rightarrow o$ and
it is greater than the type of {\tt Q} which is $\iota \rightarrow o$.\mathproofbox
\end{example}

As it turns out, stratified higher-order logic programs have an extensional well-founded
model. The proof of the following theorem can be found in~\ref{app:stratified}.
\begin{theorem}
\label{thr:strat_extens}
The well-founded model $\wfmp$ of a stratified program $\mathsf{P}$ is extensional.
\end{theorem}

Despite the fact that stratified programs lead to an extensional well-founded model,
we have not been able to verify that the same property holds for {\em locally stratified
higher-order logic programs} (for a formal definition of the notion of local stratification
for ${\cal H}$ programs, please see~\cite{RS16}). On the other hand, our attempts to find
a locally stratified program with a non-extensional well-founded model have also been
unsuccessful, and therefore it is not at present clear to us whether this class of programs
is well-behaved with respect to extensionality, or not.

\section{The Restrictions of Three-Valued Approaches}\label{restrictions}
In this section we re-examine the counterexample of Section~\ref{well_founded} but now from
a broader perspective. In particular, we indicate that in order to achieve an extensional three-valued
semantics for higher-order logic programs with negation, one has to make some (arguably)
non-standard assumptions regarding the behaviour of negation in such programs. On the other hand,
a logic with an infinite number of truth values appears to be a more appropriate
vehicle for achieving extensionality. In the following discussion, we assume some basic
familiarity with the main intuition behind the approaches described in~\cite{RS16} and
in~\cite{CharalambidisER14}.

Consider again the program of Section~\ref{well_founded}.
Under the infinite-valued adaptation of Bezem's approach given in~\cite{RS16} and also under
the domain-theoretic infinite-valued approach of~\cite{CharalambidisER14}, the semantics of
that program {\em is} extensional. The reason is that both of these approaches differentiate
the meaning of {\tt p} from the meaning of {\tt q}. The truth domain in both approaches
is the set:
$$V = \{F_\alpha\mid \alpha <\Omega\}\cup \{0\} \cup \{T_\alpha \mid \alpha<\Omega\}$$
where $F_\alpha$ and $T_\alpha$ represent different degrees of truth and falsity, and $\Omega$ is the
first uncountable ordinal. Under this truth domain, predicate {\tt p} (intuitively)
corresponds to the infinite-valued relation:
$$p=\{(v,v)\mid v \in V\}$$
while predicate {\tt q} corresponds to the relation:
$$q=\{(F_\alpha,F_{\alpha+2})\mid\alpha<\Omega\}\cup\{(0,0)\} \cup\{(T_\alpha,T_{\alpha+2})\mid \alpha<\Omega\}$$
Obviously, the relations $p$ and $q$ are different as sets and therefore it is not a surprise that under
both the semantics of~\cite{RS16} and~\cite{CharalambidisER14}, the atoms {\tt (s p)} and {\tt (s q)}
have different truth values. Notice, however, that if we collapse $p$ and $q$ in three-valued logic
(i.e., if we map each $F_\alpha$ to $\mathit{false}$, each $T_\alpha$ to $\mathit{true}$, and 0 to 0),
the collapsed relations become identical.

Assume now that want to devise an (alternative to the one presented in this paper) extensional
three-valued semantics for ${\cal H}$  programs. Under such a semantics, it seems reasonable to assume that
{\tt p} and {\tt q} would correspond to the same three-valued relation, namely $\{(v,v)\mid v\in\{\mathit{false},0,\mathit{true}\}\}$.
Notice however that from a logic programming perspective,  {\tt p} and {\tt q} are expected to have a different
operational behaviour when they appear inside a program. In particular, given the program:
\[
\begin{array}{l}
\mbox{\tt s Q $\leftarrow$ Q (s Q)}\\
\mbox{\tt p R $\leftarrow$ R}
\end{array}
\]
we expect the atom {\tt (s p)} to have the value $\mathit{false}$ (due to the circularity that occurs
if we try to evaluate it), while given the program:
\[
\begin{array}{l}
\mbox{\tt s Q $\leftarrow$ Q (s Q)}\\
\mbox{\tt q R $\leftarrow$ $\pnot$(w R)}\\
\mbox{\tt w R $\leftarrow$ $\pnot$R}
\end{array}
\]
we expect the atom {\tt (s q)} to have the value $0$ due to the circularity through negation.
At first sight, the above discussion seems to suggest that there is no way we can have a three-valued
extensional semantics for all higher-order logic programs with negation.

However, the above discussion is based mainly on our experience regarding the behaviour of first-order logic
programs with negation. One could argue that we could devise a semantics under which {\tt (s q)} will also return
the value $\mathit{false}$. One possible justification for such a semantics would be that the definition of
{\tt q} uses two negations which cancel each other, and therefore we should actually expect {\tt q} to behave exactly
like {\tt p} when it appears inside a program. Despite the fact that such a proposal seems somewhat unintuitive to
us, we can not exclude it as a possibility. It is worth noting that such cancellations of double negations
appear in certain semantic approaches to negation. For example, for certain extended propositional
programs, the semantics based on approximation fixpoint theory has the effect of canceling double
negations (see for example~\cite{DBV12}[page 185, Example 1]). It is possible that
higher-order logic programs with negation behave similarly to extended propositional programs,
and it is conceivable that one could construct an extensional three-valued semantics for all higher-order
logic programs with negation, using an approach based on approximation fixpoint theory. This research
direction certainly deserves further investigation.

It is our belief however that the most rigorous approach to extensionality for higher-order
logic programs with negation, is through the use of the infinite-valued approach. It is worth
noting that recently, some advantages of the infinite-valued approach versus the well-founded one,
were identified in a different context. More specifically, as it was recently demonstrated
in~\cite{Esik15,CarayolE16}, the infinite-valued approach satisfies all identities of
{\em iteration theories}~\cite{BloomE93}, while the well-founded semantics does not.
Since iteration theories (intuitively) provide an abstract framework for the evaluation of the merits
of various semantic approaches for languages that involve recursion, the results just mentioned give
an extra incentive for the further study and use of the infinite-valued approach.


\appendix
\section{Proof of Theorem~\ref{thr:min_mod}}\label{app:min_mod}

Recall that propositional programs consist
of clauses of the form $\mathsf{p}\leftarrow \mathsf{L}_1,\ldots,\mathsf{L}_n$, where each $\mathsf{L}_i$
is either a propositional variable or the negation of a propositional variable; the $\mathsf{L}_i$ are called {\em literals}, {\em negative} if they have negation and {\em positive} otherwise. As hinted in Section~\ref{intuitive}, we need to consider propositional programs with a possibly countably infinite number of clauses, as this is the case with the ground instantiation of a higher-order program. Moreover, we must allow a positive literal $\mathsf{L}_i$ to also be one of the constants $\mathsf{true}$ and $\mathsf{false}$. The reason for this is that in the ground instantiation of an $\mathcal{H}$ program, there may exist ground expressions of the form $(\mathsf{E}_1\approx \mathsf{E}_2)$ in the bodies of clauses. These have specific meanings under the semantics of Section~\ref{semantics_of_language} and can not be treated as propositional variables. In the case where the two expressions $\mathsf{E}_1$ and $\mathsf{E}_2$ are syntactically identical, the expression $(\mathsf{E}_1\approx \mathsf{E}_2)$ will be treated as the constant $\mathsf{true}$ (i.e., it is assumed that $I((\mathsf{E}_1\approx \mathsf{E}_2))=\mathit{true}$ for every interpretation $I$ of the ground instantiation), and otherwise as the constant $\mathsf{false}$.

We use, throughout all sections of the Appendix, the standard representation of partial interpretations of propositional programs by $\langle T, F\rangle$, where $T$ and $F$ are disjoint subsets of the Herbrand base $B_\mathsf{P}$ of a propositional program $\mathsf{P}$ (i.e., the set of propositional variables appearing in $\mathsf{P}$) denoting the sets of propositional variables considered to be true and false, respectively, in the interpretation. Naturally, $\langle T, F\rangle$ is a total interpretation if $T\cup F = B_\mathsf{P}$.
The truth ordering $\tleq$ and Fitting ordering $\fleq$ of interpretations can be defined in two equivalent ways:
\begin{definition}
If $I=\langle T, F\rangle$ and $I'=\langle T', F'\rangle$ are two partial interpretations of a propositional program $\mathsf{P}$ then we say that $I\tleq I'$ if $T\subseteq T'$ and $F'\subseteq F$, or, equivalently, if $I(\mathsf{p})\tleq I'(\mathsf{p})$ for every propositional variable $\mathsf{p}$ of $\mathsf{P}$. Moreover, we say that $I\fleq I'$ if $T\subseteq T'$ and $F\subseteq F'$, or, equivalently, if $I(\mathsf{p})\fleq I'(\mathsf{p})$ for every propositional variable $\mathsf{p}$ of $\mathsf{P}$.
\end{definition}
A model $M$ of a propositional program is as usual considered to be $\tleq$-minimal (respectively, $\fleq$-minimal) if there does not exist a different model $N$ of $\mathsf{P}$, such that $N\tleq M$ (respectively, $N\fleq M$).
\begin{oneshot}{Theorem~\ref{thr:min_mod}}
Let $\mathsf{P}$ be a program and let $\mathsf{Gr(P)}$ be its ground instantiation. Also let $M$ be a partial interpretation of $\mathsf{Gr(P)}$ and let $\mathcal{M}$ be the Herbrand interpretation of $\mathsf{P}$, such that $\vwrt{\mathsf{A}}{\mathcal{M}} = M(\mathsf{A})$ for every $\mathsf{A}\in U_{\mathsf{P},o}$.
Then, $\mathcal{M}$ is a Herbrand model of $\mathsf{P}$ if and only if $M$ is a model of $\mathsf{Gr(P)}$. Moreover, $\mathcal{M}$ is $\tleq$-minimal (respectively, $\fleq$-minimal) if and only if $M$ is $\tleq$-minimal (respectively, $\fleq$-minimal).
\end{oneshot}
\begin{proof*}
\begin{description}
\item[Step 1] $M$ is a model of $\mathsf{Gr(P)}\Rightarrow\mathcal{M}$ is a Herbrand model of $\mathsf{P}$:
For every Herbrand state $s$ of $\mathsf{P}$ there exists a ground substitution $\theta$ such that $\theta(\mathsf{V}) = s(\mathsf{V})$, and therefore $s(\mathsf{V}) = \mwrs{\theta(\mathsf{V})}{\mathcal{M}}{s'}$, for all states $s'$ and variables $\mathsf{V}$ in $\mathsf{P}$. Also, for every clause $\mathsf{A} \leftarrow \mathsf{L}_1, \ldots, \mathsf{L}_m$ in $\mathsf{P}$ there exists a respective ground instance $\mathsf{A}\theta \leftarrow \mathsf{L}_1\theta, \ldots, \mathsf{L}_m\theta$ in $\mathsf{Gr(P)}$. As $M$ is a model of $\mathsf{Gr(P)}$, $M(\mathsf{A}\theta)\geq min\{M(\mathsf{L}_1\theta), \ldots, M(\mathsf{L}_m\theta)\}$. By assumption, $\vwrt{\mathsf{A}\theta}{\mathcal{M}}= M(\mathsf{A}\theta)$ and $\vwrt{\mathsf{L}_i\theta}{\mathcal{M}}= M(\mathsf{L}_i\theta)$ for all $i\leq m$. Moreover, it is easy to see (by a trivial induction on the structure of the expression) that $\mathsf{A}\theta= \mwrs{\mathsf{A}}{\mathcal{M}}{s}$, which implies that $\vwrt{\mathsf{A}\theta}{\mathcal{M}}= \vwrt{\mwrs{\mathsf{A}}{\mathcal{M}}{s}}{\mathcal{M}}$. Similarly, $\vwrt{\mathsf{L}_i\theta}{\mathcal{M}}= \vwrt{\mwrs{\mathsf{L}_i}{\mathcal{M}}{s}}{\mathcal{M}}$, for all $i\leq m$. It follows immediately that $\vwrt{\mwrs{\mathsf{A}}{\mathcal{M}}{s}}{\mathcal{M}}\geq min\{\vwrt{\mwrs{\mathsf{L}_1}{\mathcal{M}}{s}}{\mathcal{M}}, \ldots, \vwrt{\mwrs{\mathsf{L}_m}{\mathcal{M}}{s}}{\mathcal{M}}\}$, which implies that $\mathcal{M}$ is a model of $\mathsf{P}$.
\item[Step 2]
$\mathcal{M}$ is a Herbrand model of $\mathsf{P}\Rightarrow M$ is a model of $\mathsf{Gr(P)}$:
Every clause in $\mathsf{Gr(P)}$ is a ground instance of a clause $\mathsf{A} \leftarrow \mathsf{L}_1, \ldots, \mathsf{L}_m$ in $\mathsf{P}$ and is therefore of the form $\mathsf{A}\theta \leftarrow \mathsf{L}_1\theta, \ldots, \mathsf{L}_m\theta$ for some ground substitution $\theta$. Consider a Herbrand state $s$, such that $s(\mathsf{V}) = \theta(\mathsf{V})$ for every variable $\mathsf{V}$ in $\mathsf{P}$. Because $\mathcal{M}$ is a model of $\mathsf{P}$, we have that $\vwrt{\mwrs{\mathsf{A}}{\mathcal{M}}{s}}{\mathcal{M}}\geq min\{\vwrt{\mwrs{\mathsf{L}_1}{\mathcal{M}}{s}}{\mathcal{M}}, \ldots, \vwrt{\mwrs{\mathsf{L}_m}{\mathcal{M}}{s}}{\mathcal{M}}\}$. Again, it is easy to see that $\mathsf{A}\theta= \mwrs{\mathsf{A}}{\mathcal{M}}{s}$ and therefore $\vwrt{\mathsf{A}\theta}{\mathcal{M}}= \vwrt{\mwrs{\mathsf{A}}{\mathcal{M}}{s}}{\mathcal{M}}$. Similarly, $\vwrt{\mathsf{L}_i\theta}{\mathcal{M}}= \vwrt{\mwrs{\mathsf{L}_i}{\mathcal{M}}{s}}{\mathcal{M}}$ for all $i\leq m$. Additionally, $\vwrt{\mathsf{A}\theta}{\mathcal{M}}= M(\mathsf{A}\theta)$ and $\vwrt{\mathsf{L}_i\theta}{\mathcal{M}}= M(\mathsf{L}_i\theta)$ for all $i\leq m$, so $M(\mathsf{A}\theta)\geq min\{M(\mathsf{L}_1\theta), \ldots, M(\mathsf{L}_m\theta)\}$, which implies that $M$ is a model of $\mathsf{Gr(P)}$.
\item[Step 3]
$M$ is minimal $\Rightarrow\mathcal{M}$ is minimal:
Assume there exists a model $\mathcal{N}$ of $\mathsf{P}$, distinct from $\mathcal{M}$, such that $\mathcal{N}\tleq \mathcal{M}$ (respectively, $\mathcal{N}\fleq \mathcal{M}$).
Then we can construct an interpretation $N$ for $\mathsf{Gr(P)}$ such that for every ground atom $\mathsf{A}$, $N(\mathsf{A}) = \vwrt{\mathsf{A}}{\mathcal{N}}$. It is obvious that $N\tleq M$ (respectively, $N\fleq M$), since $N(\mathsf{A}) = \vwrt{\mathsf{A}}{\mathcal{N}}\tleq\vwrt{\mathsf{A}}{\mathcal{M}} = M(\mathsf{A})$ (respectively,  $\vwrt{\mathsf{A}}{\mathcal{N}}\fleq\vwrt{\mathsf{A}}{\mathcal{M}}$).
Also, $N$ is distinct from $M$, since $N(\mathsf{B}) = \vwrt{\mathsf{B}}{\mathcal{N}} \neq \vwrt{\mathsf{B}}{\mathcal{M}} = M(\mathsf{B})$ for at least one ground atom $\mathsf{B}$. As we showed in Step 2, the fact that $\mathcal{N}$ is a model of $\mathsf{P}$ implies that $N$ is a model of $\mathsf{Gr(P)}$, which is of course a contradiction, since $M$ is a $\tleq$-minimal (respectively, $\fleq$-minimal) model of $\mathsf{Gr(P)}$. Therefore, $\mathcal{M}$ must be a $\tleq$-minimal (respectively, $\fleq$-minimal) model of $\mathsf{P}$.
\item[Step 4]
$\mathcal{M}$ is minimal $\Rightarrow M$ is minimal:
By the reverse of the argument used in Step 3:
Assume there exists a model $N$ of $\mathsf{Gr(P)}$, distinct from $M$, such that $N\tleq M$ (respectively, $N\fleq M$).
Then we can construct an interpretation $\mathcal{N}$ for $\mathsf{P}$ such that for every ground atom $\mathsf{A}$, $N(\mathsf{A}) = \vwrt{\mathsf{A}}{\mathcal{N}}$. It is obvious that $\mathcal{N}\tleq \mathcal{M}$ (respectively, $\mathcal{N}\fleq \mathcal{M}$), since $\vwrt{\mathsf{A}}{\mathcal{N}}=N(\mathsf{A})\tleq M(\mathsf{A})=\vwrt{\mathsf{A}}{\mathcal{M}}$ (respectively, $N(\mathsf{A}) \fleq M(\mathsf{A})$).
Also, $\mathcal{N}$ is distinct from $\mathcal{M}$, since their valuation functions are distinct. As we showed in Step 1, the fact that $N$ is a model of $\mathsf{Gr(P)}$ implies that $\mathcal{N}$ is a model of $\mathsf{P}$, which is of course a contradiction, since $\mathcal{M}$ is a $\tleq$-minimal (respectively, $\fleq$-minimal) model of $\mathsf{P}$. Therefore, $M$ must be a $\tleq$-minimal (respectively, $\fleq$-minimal) model of $\mathsf{Gr(P)}$.\mathproofbox
\end{description}
\end{proof*}

\section{Proof of Lemma~\ref{lm:exm_is_not_extens}}\label{app:counterexample}
For the proof of Lemma~\ref{lm:exm_is_not_extens}, which we present in this appendix, we rely on the method of~\cite{PrzP90} for the construction of the well-founded model. We first give the necessary definitions from~\cite{PrzP90}.
\begin{definition}
Let $\mathsf{P}$ be a propositional program and let $J$ be an interpretation of $\mathsf{P}$. The operator $\Theta_J(\cdot)$ on the set of interpretations of $\mathsf{P}$ is defined as follows: for every interpretation $I$ and every propositional variable $\mathsf{p}$ of $\mathsf{P}$,
\[
\Theta_J(I)(\mathsf{p}) = \begin{cases}
\mathit{true}, & \mbox{there exists a clause } \mathsf{p}\leftarrow \mathsf{L}_1,\ldots,\mathsf{L}_n \mbox{ in } \mathsf{P} \mbox{ s.t. for all } i\leq n,\\
& \mbox{either } J(\mathsf{L}_i)= \mathit{true} \mbox{ or } \mathsf{L}_i \mbox{ is a positive literal and } I(\mathsf{L}_i)=\mathit{true};\\
\mathit{false}, & \mbox{for all clauses } \mathsf{p}\leftarrow \mathsf{L}_1,\ldots,\mathsf{L}_n \mbox{ in } \mathsf{P} \mbox{ there exists an } i\leq n, \mbox{ s.t.}\\
& \mbox{either } J(\mathsf{L}_i)= \mathit{false} \mbox{ or } \mathsf{L}_i \mbox{ is a positive literal and } I(\mathsf{L}_i)=\mathit{false};\\
\unk, & \mbox{otherwise}.

\end{cases}
\]
Moreover we define the following sequence of interpretations:
\[
\begin{array}{lll}
\Theta_J^{\uparrow 0} &  = & \langle T_0, F_0\rangle=\langle \emptyset, B_\mathsf{P}\rangle\\
\Theta_J^{\uparrow (n+1)} &  = & \langle T_{n+1}, F_{n+1}\rangle=\Theta_J(\Theta_J^{\uparrow n})\\
\Theta_J^{\uparrow \omega} &  = & \langle T_{\omega}, F_{\omega}\rangle=\langle\bigcup_{n<\omega} T_\beta,\bigcap_{n<\omega} F_\beta\rangle
\end{array}
\]
\end{definition}
It is shown in~\cite{PrzP90} that, for any interpretation $J$, the operator $\Theta_J$ has a unique least fixed-point given by $\Theta_J^{\uparrow \omega}$.
\begin{definition}
Let $\mathsf{P}$ be a propositional program. For every countable ordinal $\alpha\leq\gamma$, we define the interpretation $M_\alpha$ as follows:
\[
\begin{array}{lll}
M_0 &  = & \langle T_0, F_0\rangle=\langle \emptyset, \emptyset\rangle\\
M_{\alpha+1} &  = & \langle T_{\alpha+1}, F_{\alpha+1}\rangle=\Theta_{M_\alpha}^{\uparrow \omega}, \mbox{ for a successor ordinal } \alpha+1\\
M_\alpha &  = & \langle T_{\alpha}, F_{\alpha}\rangle=\langle\bigcup_{\beta<\alpha} T_\beta,\bigcup_{\beta<\alpha} F_\beta\rangle, \mbox{ for a limit ordinal } \alpha
\end{array}
\]
\end{definition}
Again from~\cite{PrzP90}, there exists the least countable ordinal $\lambda$, such that $M_{\lambda}=\Theta_{M_{\lambda}}^{\uparrow \omega}$ and $M_{\lambda}$ coincides with the well-founded model $M_\mathsf{P}$ of the propositional program $\mathsf{P}$.
We now present the proof of Lemma~\ref{lm:exm_is_not_extens}.
\begin{oneshot}{Lemma~\ref{lm:exm_is_not_extens}}
The Herbrand interpretation $\wfmp$ of the program of Example~\ref{exm:not_extens} is not extensional.
\end{oneshot}
\begin{proof}
We repeat here the program of Example~\ref{exm:not_extens} for the reader's convenience:
\[
\begin{array}{l}
\mbox{\tt s Q $\leftarrow$ Q (s Q)}\\
\mbox{\tt p R $\leftarrow$ R}\\
\mbox{\tt q R $\leftarrow$ $\pnot$(w R)}\\
\mbox{\tt w R $\leftarrow$ $\pnot$R}
\end{array}
\]
Recall that the predicate variable {\tt Q} is of type $o\rightarrow o$ and the predicate variable {\tt R}
is of type $o$. The ground instantiation of the above program is infinite, as so:
%
\[
\begin{array}{l}
\mbox{\tt s p $\leftarrow$ p (s p)}\\
\mbox{\tt s q $\leftarrow$ q (s q)}\\
\mbox{\tt s w $\leftarrow$ w (s w)}\\
\mbox{\tt p (s p) $\leftarrow$ (s p)}\\
\mbox{\tt p (s q) $\leftarrow$ (s q)}\\
\mbox{\tt p (s w) $\leftarrow$ (s w)}\\
\mbox{\tt q (s p) $\leftarrow \pnot$(w (s p))}\\
\mbox{\tt w (s p) $\leftarrow \pnot$(s p)}\\
\mbox{\tt q (s q) $\leftarrow \pnot$(w (s q))}\\
\mbox{\tt w (s q) $\leftarrow \pnot$(s q)}\\
\mbox{\tt q (s w) $\leftarrow \pnot$(w (s w))}\\
\mbox{\tt w (s w) $\leftarrow \pnot$(s w)}\\
\ldots
\end{array}
\]
The well-founded model $\wgp$ of the above (infinite) propositional program is the valuation function of $\wfmp$. It has already been argued in Section~\ref{well_founded} and Section~\ref{restrictions}, that $\wgp(\mbox{\tt s p})\neq\wgp(\mbox{\tt s q})$ and this should be obvious to the reader who is familiar with the well-founded model semantics; however, for reasons of completeness, we present a formal argument in the second part of this proof. In the same sections, we claimed that this is despite the fact that $\mbox{\tt p} \exeq[\wgp,o\rightarrow o] \mbox{\tt q}$, which we will immediately proceed to prove. Of course, by Definitions~\ref{def:extensional_eq} and~\ref{def:extensionality}, the facts that $\wgp(\mbox{\tt s p})\neq\wgp(\mbox{\tt s q})$ and $\mbox{\tt p} \exeq[\wgp,o\rightarrow o] \mbox{\tt q}$, render $\wfmp$ not extensional.

First, we show that $\mbox{\tt p} \exeq[\wgp,o\rightarrow o] \mbox{\tt q}$, i.e. that for all $\mathsf{A}, \mathsf{A'}\in U_{\mathsf{P},o}$ such that $\mathsf{A}\exeq[\wgp,o]\mathsf{A'}$, $\mbox{\tt p } \mathsf{A}\exeq[\wgp,o] \mbox{\tt q } \mathsf{A'}$ holds. By definition $\wgp$ is a fixed-point of the operator $\Theta_{\wgp}(\cdot)$, therefore for any ground atom $\mathsf{B}$ we have that $\wgp(\mathsf{B})$ equals to $\mathit{true}$, if there exists a clause $\mathsf{B}\leftarrow \mathsf{L}_1, \ldots, \mathsf{L}_n$ in $\mathsf{Gr(P)}$, such that $\wgp(\mathsf{L}_i)=\mathit{true}$ for all $i\leq n$; it equals to $\mathit{false}$ if for every clause $\mathsf{B}\leftarrow \mathsf{L}_1, \ldots, \mathsf{L}_n$ in $\mathsf{Gr(P)}$, we have that $\wgp(\mathsf{L}_i)=\mathit{false}$ for at least one $i\leq n$; and it equals to $\unk$ otherwise. Observe that there exists only one clause in $\mathsf{Gr(P)}$ such that $\mbox{\tt p }\mathsf{A}$ is the head of the clause, in particular it is the ground instance $\mbox{\tt p }\mathsf{A}\leftarrow \mathsf{A}$ of the clause $\mbox{\tt p R} \leftarrow \mbox{\tt R}$ of $\mathsf{P}$. This suggests that $\wgp(\mbox{\tt p }\mathsf{A})=\wgp(\mathsf{A})$~(1). Similarly, the ground instance $\mbox{\tt q }\mathsf{A'}\leftarrow\pnot\mbox{\tt (w }\mathsf{A'}\mbox{\tt )}$ of the clause $\mbox{\tt q R}\leftarrow\pnot\mbox{\tt (w R)}$ is the only clause in $\mathsf{Gr(P)}$ with $\mbox{\tt q }\mathsf{A'}$ as its head atom and from this we can infer that $\wgp(\mbox{\tt q }\mathsf{A'})=\wgp(\pnot\mbox{\tt (w }\mathsf{A'}\mbox{\tt )})=\lnot\wgp(\mbox{\tt w }\mathsf{A'})$~(2). Finally, the only clause in $\mathsf{Gr(P)}$, such that $\mbox{\tt w }\mathsf{A'}$ is the head of the clause, is the ground instance $\mbox{\tt w }\mathsf{A'}\leftarrow\pnot\mathsf{A'}$ of the clause $\mbox{\tt w R}\leftarrow\pnot\mbox{\tt R}$ of $\mathsf{P}$, which implies that $\wgp(\mbox{\tt w }\mathsf{A'})=\wgp(\pnot\mathsf{A'})=\lnot\wgp(\mathsf{A'})$~(3). By (2) and (3) we have $\wgp(\mbox{\tt q }\mathsf{A'})=\wgp(\mathsf{A'})$, and, in conjunction with (1), that $\wgp(\mbox{\tt p }\mathsf{A})=\wgp(\mbox{\tt q }\mathsf{A'})$, because $\mathsf{A}\exeq[\wgp,o]\mathsf{A'}$ implies, by Definition~\ref{def:extensional_eq}, that $\wgp(\mathsf{A})=\wgp(\mathsf{A'})$. Therefore, we also have $\mbox{\tt p }\mathsf{A}\exeq[\wgp,o]\mbox{\tt q }\mathsf{A'}$ and, in consequence, $\mbox{\tt p}\exeq[\wgp,o\rightarrow o]\mbox{\tt q}$.

For the second part, we show that $\wgp(\mbox{\tt s p})\neq\wgp(\mbox{\tt s q})$. We do this in two steps; first we show that $\wgp(\mbox{\tt s p})=\mathit{false}$ and then that $\wgp(\mbox{\tt s q})=\unk$.

For the first step, it suffices to show that ${M_1}(\mbox{\tt s p})=\Theta^{\uparrow \omega}_{M_0}(\mbox{\tt s p})=\mathit{false}$. For this, we prove that $\Theta^{\uparrow n}_{M_0}(\mbox{\tt s p})=\mathit{false}$ and $\Theta^{\uparrow n}_{M_0}(\mbox{\tt p (s  p)})=\mathit{false}$, for all $n<\omega$, by an induction on $n$. The basis case is trivial, as $\Theta^{\uparrow 0}_{M_0}=\langle\emptyset,U_{\mathsf{P},o}\rangle$, assigns the value $\mathit{false}$ to every atom. For the induction step, we show the statement for $n+1$ assuming that it holds for $n$. We see that there exists only one clause in $\mathsf{Gr(P)}$ such that {\tt s p} is the head of the clause; this is the ground instance $\mbox{\tt s p}\leftarrow \mbox{\tt p (s p)}$ of the clause $\mbox{\tt s Q}\leftarrow \mbox{\tt Q (s Q)}$ of $\mathsf{P}$. By the induction hypothesis, we have that $\Theta^{\uparrow n}_{M_0}(\mbox{\tt p (s p)})=\mathit{false}$, therefore $\Theta^{\uparrow (n+1)}_{M_0}(\mbox{\tt s p})=\mathit{false}$. Similarly, the only clause in $\mathsf{Gr(P)}$ with {\tt p (s p)} as the head of the clause is the ground instance $\mbox{\tt p (s p)}\leftarrow \mbox{\tt (s p)}$ of the clause $\mbox{\tt p R}\leftarrow \mbox{\tt R}$ of $\mathsf{P}$. By the induction hypothesis, we have that $\Theta^{\uparrow n}_{M_0}(\mbox{\tt s p})=\mathit{false}$, therefore $\Theta^{\uparrow (n+1)}_{M_0}(\mbox{\tt p (s p)})=\mathit{false}$.

For the second step, we perform an induction on $\alpha$, during which we simultaneously show that ${M_\alpha}(\mbox{\tt s q})=\unk$, ${M_\alpha}(\mbox{\tt q (s q)})=\unk$ and ${M_\alpha}(\mbox{\tt w (s q)})=\unk$, for all countable ordinals $\alpha$. The basis case is trivial, as $M_0=\langle\emptyset,\emptyset\rangle$ assigns the value $\unk$ to all atoms. For the induction step, we first prove the statement for a successor ordinal $\alpha+1$, assuming that it holds for all countable ordinals up to $\alpha$. Indeed, there exists exactly one clause in $\mathsf{Gr(P)}$ with {\tt w (s q)} as its head atom, in particular the ground instance $\mbox{\tt w (s q)}\leftarrow\pnot\mbox{\tt (s q)}$ of the clause $\mbox{\tt q R}\leftarrow\pnot\mbox{\tt (w R)}$. As $\pnot\mbox{\tt (s q)}$ is a negative literal, for every $n<\omega$ the value of $\Theta_{M_\alpha}^{\uparrow (n+1)}(\mbox{\tt w (s q)})$ is defined by $M_\alpha(\pnot\mbox{\tt (s q)})$. By the induction hypothesis, we have $M_\alpha(\mbox{\tt s q})=M_\alpha(\pnot\mbox{\tt (s q)})=\unk$, therefore it follows that $\Theta^{\uparrow (n+1)}_{M_\alpha}(\mbox{\tt w (s q)})=\unk$ and, because this holds for every $n<\omega$, that ${M_{\alpha+1}}(\mbox{\tt w (s q)})=\unk$. Moreover, the ground instance $\mbox{\tt q (s q)}\leftarrow\pnot\mbox{\tt (w (s q))}$ of the clause $\mbox{\tt q R}\leftarrow\pnot\mbox{\tt (w R)}$ is the only clause in $\mathsf{Gr(P)}$ with {\tt q (s q)} as its head atom. Again, $\pnot\mbox{\tt (w (s q))}$ is a negative literal and so for every $n<\omega$ the value of $\Theta^{\uparrow (n+1)}_{M_\alpha}(\mbox{\tt q (s q)})$ only depends on $M_\alpha(\pnot\mbox{\tt (w (s q))})$. By the induction hypothesis, we have $M_\alpha(\mbox{\tt w (s q)})=M_\alpha(\pnot\mbox{\tt (w (s q))})=\unk$, therefore it follows that $\Theta^{\uparrow (n+1)}_{M_\alpha}(\mbox{\tt q (s q)})=\unk$. Since this holds for every $n<\omega$, we also have that ${M_{\alpha+1}}(\mbox{\tt q (s q)})=\unk$. Finally,  there exists only one clause in $\mathsf{Gr(P)}$ such that {\tt s q} is the head of the clause, in particular the ground instance $\mbox{\tt s q}\leftarrow \mbox{\tt q (s q)}$ of the clause $\mbox{\tt s Q}\leftarrow \mbox{\tt Q (s Q)}$ of $\mathsf{P}$. We have already shown that $\Theta^{\uparrow (n+1)}_{M_\alpha}(\mbox{\tt q (s q)})=\unk$ for all $n<\omega$; moreover, by the induction hypothesis, ${M_\alpha}(\mbox{\tt q (s q)})=\unk$. Consequently, for all $n<\omega$, $\Theta^{\uparrow (n+2)}_{M_\alpha}(\mbox{\tt s q})=\unk$ and thus $M_{\alpha+1}(\mbox{\tt s q})=\unk$. It remains to show ${M_\alpha}(\mbox{\tt s q})=\unk$, ${M_\alpha}(\mbox{\tt q (s q)})=\unk$ and ${M_\alpha}(\mbox{\tt w (s q)})=\unk$ for a limit ordinal $\alpha$. In this case, we have that $M_\alpha=\langle\bigcup_{\beta<\alpha} T_\beta,\bigcup_{\beta<\alpha} F_\beta\rangle$. By the induction hypothesis, ${M_\beta}(\mbox{\tt s q})=\unk$ for all $\beta<\alpha$, which means that $\mbox{\tt s q}\not\in T_\beta$ and $\mbox{\tt s q}\not\in F_\beta$. In other words, $\mbox{\tt s q}\not\in \bigcup_{\beta<\alpha} T_\beta$ and $\mbox{\tt s q}\not\in \bigcup_{\beta<\alpha} F_\beta$, therefore $M_\alpha(\mbox{\tt s q})=\unk$. In the same way we can show that ${M_\alpha}(\mbox{\tt q (s q)})=\unk$ and ${M_\alpha}(\mbox{\tt w (s q)})=\unk$. This concludes the induction and so we have proven that $\wgp(\mbox{\tt s q})=\unk$.
\end{proof}

\section{Proof of Theorem~\ref{thr:strat_extens}}\label{app:stratified}
%
Before we proceed with the proof of Theorem~\ref{thr:strat_extens}, we recall some necessary definitions. Note that the proof makes use of a fixed-point characterization of the perfect model semantics given in~\cite{PrzP90}, rather than the more traditional definition of~\cite{Prz88}. The following definition of the local stratification of possibly infinite propositional programs is adapted to allow for the presence of expressions of the form $(\mathsf{E}_1\approx\mathsf{E}_2)$.
\begin{definition}
A propositional program $\mathsf{P}$ is called {\em locally stratified} if and only if it is possible to decompose the
Herbrand base $B_{\mathsf{P}}$ of $\mathsf{P}$ into disjoint sets (called {\em strata})
$S_1,S_2,\ldots,S_\alpha,\ldots, \alpha<\gamma$, where $\gamma$ is a countable ordinal, such that
for every clause $\mathsf{H} \leftarrow \mathsf{A}_1,\ldots,\mathsf{A}_m,\pnot \mathsf{B}_1,\ldots,\pnot \mathsf{B}_n$
in $\mathsf{P}$, we have that for every $i\leq m$, $\textit{stratum}(\mathsf{A}_i) \leq \textit{stratum}(\mathsf{H})$ and
for every $i\leq n$, $\textit{stratum}(\mathsf{B}_i) < \textit{stratum}(\mathsf{H})$, where $\textit{stratum}$ is a function
such that $\textit{stratum}(\mathsf{C}) = \beta$,
if the propositional variable $\mathsf{C}\in B_{\mathsf{P}}$ belongs to $S_\beta$ and $\textit{stratum}(\mathsf{C}) = 0$, if  $\mathsf{C}\not\in B_{\mathsf{P}}$ and is a constant (equivalently, of the form $(\mathsf{E}_1\approx\mathsf{E}_2)$). Any decomposition of the described form is called a \emph{local stratification} of $\mathsf{P}$.
\end{definition}
In \cite{PrzP90}, for any interpretation $J$, the operator $\Psi_J$ is defined and shown to have a unique least fixed-point given by $\Psi_J^{\uparrow \omega}$ of the next definition. This is then used to give an iterated fixed-point characterization of the perfect model of a locally stratified program.
\begin{definition}
Let $\mathsf{P}$ be a propositional program and let $J$ be an interpretation of $\mathsf{P}$. The operator $\Psi_J : 2^{B_\mathsf{P}}\rightarrow 2^{B_\mathsf{P}}$ is defined as follows: for every $I\subseteq {B_\mathsf{P}}$,
$\Psi_J(I) = \{\mathsf{p}\in {B_\mathsf{P}}\mid \mbox{there exists a clause } \mathsf{p}\leftarrow \mathsf{L}_1,\ldots,\mathsf{L}_n \mbox{ in } \mathsf{P} \mbox{ such that, for all } i\leq n \mbox{, either } J(\mathsf{L}_i)= \mathit{true} \mbox{ or}$ $\mathsf{L}_i\in I \}$.
Moreover we define the following sequence:
\[
\begin{array}{lll}
\Psi_J^{\uparrow 0} &  = & \emptyset\\
\Psi_J^{\uparrow (n+1)} &  = & \Psi_J(\Psi_J^{\uparrow n})\\
\Psi_J^{\uparrow \omega} &  = & \bigcup_{n<\omega} \Psi_J^{\uparrow n}
\end{array}
\]
\end{definition}
We follow the usual convention of identifying a subset of the Herbrand base with a total (two-valued) interpretation of the program and use the two notions interchangeably. E.g., a set $\Psi_J^{\uparrow n}$ of the above sequence is considered to be equivalent to the interpretation $\langle \Psi_J^{\uparrow n},B_\mathsf{P}-\Psi_J^{\uparrow n}\rangle$ and $\mathsf{p}\in\Psi_J^{\uparrow n}$ is considered to be equivalent to $\Psi_J^{\uparrow n}(\mathsf{p})=\mathit{true}$.
Given a local stratification $S_1,S_2,\ldots,S_\alpha,\ldots$, $\alpha<\gamma$, of a propositional program $\mathsf{P}$, we define the sets $\mathcal{B}_\alpha = \bigcup_{\beta<\alpha} S_\beta$ for every countable ordinal $\alpha\leq\gamma$. Clearly, $B_\mathsf{P}=\mathcal{B}_\gamma$. Then the perfect model of $\mathsf{P}$ can be constructed~\cite{PrzP90} as the last interpretation $N_\gamma$ in an $\fleq$-increasing sequence of partial interpretations of $\mathsf{P}$:
\begin{definition}
Let $\mathsf{P}$ be a propositional program and let $S_1,S_2,\ldots,S_\alpha,\ldots, \alpha<\gamma$, where $\gamma$ is a countable ordinal, be a local stratification of $\mathsf{P}$. For every countable ordinal $\alpha\leq\gamma$, we define the interpretation $N_\alpha$ as follows:
\[
\begin{array}{lll}
N_0 &  = & \langle T_0, F_0\rangle=\langle \emptyset, \emptyset\rangle\\
N_{\alpha+1} &  = & \langle T_{\alpha+1}, F_{\alpha+1}\rangle=\langle\Psi_{N_\alpha}^{\uparrow \omega}, \mathcal{B}_{\alpha+1}-\Psi_{N_\alpha}^{\uparrow \omega}\rangle, \mbox{ for a successor ordinal } \alpha+1\\
N_\alpha &  = & \langle T_{\alpha}, F_{\alpha}\rangle=\langle\bigcup_{\beta<\alpha} T_\beta,\bigcup_{\beta<\alpha} F_\beta\rangle, \mbox{ for a limit ordinal } \alpha
\end{array}
\]
\end{definition}
\removebrackets\begin{theorem}[\cite{PrzP90}]\label{thr:perf_model_char}
Let $\mathsf{P}$ be a propositional program. The sequence $N_0,N_1,\ldots,N_\alpha,\ldots, N_\gamma$ is $\fleq$-increasing. Moreover, $N_\gamma$ coincides with the perfect model $N_\mathsf{P}$ of $\mathsf{P}$.
\end{theorem}
The next lemma and its following corollary are the basis for our proof of Theorem~\ref{thr:strat_extens}.
\begin{lemma}\label{lem:strat_to_localstr}
Let $\mathsf{P}$ be an $\mathcal{H}$ program. If $\mathsf{P}$ is stratified then the ground instantiation $\mathsf{Gr(P)}$ of $\mathsf{P}$ is locally stratified.
\end{lemma}
\begin{proof}
Consider a decomposition $S_1,\ldots,S_r$ of the set of predicate constants of $\mathsf{P}$ such that the
requirements of Definition~\ref{stratified} are satisfied. This defines a
decomposition $S'_1,\ldots,S'_r$ of $U_{\mathsf{P},o}$, which is also the Herbrand base of $\mathsf{Gr(P)}$, as follows:
$$S'_i = \{\mathsf{A} \in U_{\mathsf{P},o} \mid \mbox{the leftmost predicate constant of $\mathsf{A}$ belongs to $S_i$}\}$$
It is easy to check that $S'_1,\ldots,S'_r$ corresponds to a local stratification of $\mathsf{Gr(P)}$.
\end{proof}
An immediate result of the above lemma is that the model $\pmp$ can be defined for every stratified program of $\mathcal{H}$:
\begin{corollary}\label{crl:perfect_mod_exists}
Let $\mathsf{P}$ be an $\mathcal{H}$ program. If $\mathsf{P}$ is stratified, then the perfect model $\pmp$ of $\mathsf{P}$ exists and coincides with its well-founded model $\wfmp$.
\end{corollary}
\begin{proof}
By Lemma~\ref{lem:strat_to_localstr}, if $\mathsf{P}$ is stratified then $\mathsf{Gr(P)}$ is locally stratified. Therefore the unique perfect model $\pgp$ of $\mathsf{Gr(P)}$ exists~\cite{Prz88} and $\pmp$ is defined. Moreover, the perfect model of a locally stratified propositional program, which is the valuation function of $\pmp$, coincides with its well-founded model $\wgp$~\cite{Prz88}, i.e. the valuation function of $\wfmp$. In other words, in this case $\pmp$ and $\wfmp$ coincide, because they have the same valuation function.
\end{proof}
\begin{oneshot}{Theorem~\ref{thr:strat_extens}}
The well-founded model $\wfmp$ of a stratified program $\mathsf{P}$ is extensional.
\end{oneshot}
\begin{proof}
By Corollary~\ref{crl:perfect_mod_exists}, if $\mathsf{P}$ is stratified then $\wfmp$ coincides with $\pmp$. Therefore, it suffices to show that $\pmp$ is extensional and for this we rely upon the constructive definition of $\pgp$ from~\cite{PrzP90} presented above.

Consider a stratification $S_1,\ldots,S_r$ of the set of predicate constants of $\mathsf{P}$. As argued in the proof of Lemma~\ref{lem:strat_to_localstr}, the following
decomposition $S'_1,\ldots,S'_r$ of $U_{\mathsf{P},o}$:
\[S'_i = \{\mathsf{A} \in U_{\mathsf{P},o} \mid \mbox{the leftmost predicate constant of $\mathsf{A}$ belongs to $S_i$}\}\]
corresponds to a local stratification of $\mathsf{Gr(P)}$. Therefore, whenever a ground atom $\mathsf{A}$ begins with a predicate constant $\mathsf{p}$, we will have $stratum(\mathsf{A})=stratum(\mathsf{p})$. Moreover, by Theorem~\ref{thr:perf_model_char}, $\pgp=N_{r}$.

Since the valuation function of $\pmp$ is $\pgp$, essentially we need to show that $\mathsf{E}\exeq[\pgp,\rho] \mathsf{E}$, for every ground expression $\mathsf{E}$ of every argument type $\rho$.
We perform an induction on the structure of $\rho$. For the base types $\iota$ and $o$ the statement holds by definition. 
For the induction step, we prove the statement for a predicate type $\pi=\rho_1\rightarrow\cdots\rightarrow\rho_m\rightarrow o$, assuming that it holds for all types simpler than $\pi$ (i.e., for the types $\rho_1,\ldots,\rho_m,o$ and, recursively, the types that are simpler than $\rho_1,\ldots,\rho_m$).
Let $\mathsf{A}$ be any atom of the following form: $\mathsf{A}$ is headed by a predicate constant $\mathsf{p}$ and all variables in $vars(\mathsf{A})$ are of types simpler than $\pi$. Let $\theta, \theta'$ be  ground substitutions, such that $vars(\mathsf{A}) \subseteq dom(\theta),dom(\theta')$ and $\theta(\mathsf{V}) \exeq[\pgp,\rho] \theta'(\mathsf{V})$ for any $\mathsf{V}:\rho$ in $vars(\mathsf{A})$. We claim it suffices to show the following two properties $P_1(\alpha)$ and $P_2(\alpha)$, for all finite ordinals (i.e., natural numbers) $\alpha$:
\begin{description}
\item[$P_1(\alpha)$:] If $\vgr{\mathsf{A}\theta}{N_\alpha}=\mathit{true}$ then $\vgr{\mathsf{A}\theta'}{\pgp}=\mathit{true}$.
\item[$P_2(\alpha)$:] If $\vgr{\mathsf{A}\theta}{N_{\alpha}}=\mathit{false}$ then $\vgr{\mathsf{A}\theta'}{\pgp}=\mathit{false}$.
\end{description}
To see why proving the above properties is enough to establish that $\mathsf{E}\exeq[\pgp, \pi] \mathsf{E}$, observe the following: first of all, we assumed that $\pi$ is of the form $\rho_1\rightarrow\cdots\rightarrow\rho_m\rightarrow o$, so if $\mathsf{V}_1:\rho_1, \ldots, \mathsf{V}_m:\rho_m$ are variables, then $\mathsf{E}\, \mathsf{V}_1\, \cdots\, \mathsf{V}_m$ is an atom of the form described above. As $\pgp=N_{r}$, if $\vgr{\mathsf{E}\,\theta(\mathsf{V}_1)\,\cdots\,\theta(\mathsf{V}_m)}{\pgp}=\vgr{\mathsf{E}\,\theta(\mathsf{V}_1) \,\cdots\,\theta(\mathsf{V}_m)}{N_r}=\mathit{true}$ and property $P_1(r)$ holds, then we can infer that $\vgr{\mathsf{E}\,\theta'(\mathsf{V}_1) \,\cdots\, \theta'(\mathsf{V}_m)}{\pgp}=\mathit{true}$.
Because the relations $\exeq[\pgp,\rho_i]$ are symmetric, $\theta$ and $\theta'$ are interchangeable. Therefore the same argument can be used to infer the reverse implication, i.e. $\vgr{\mathsf{E}\, \theta'(\mathsf{V}_1) \, \cdots\, \theta'(\mathsf{V}_m)}{\pgp}=\mathit{true}\Rightarrow\vgr{\mathsf{E}\, \theta(\mathsf{V}_1) \, \cdots\, \theta(\mathsf{V}_m)}{\pgp}=\mathit{true}$, thus resulting to an equivalence. If $P_2(r)$ holds, the analogous equivalence can be shown for the value $\mathit{false}$ in the same way and so it follows that $\mathsf{E}\exeq[\pgp, \pi] \mathsf{E}$. Finally, $r$ is determined by the stratification of the higher-order program and is therefore finite, so we only need to prove properties $P_1(\alpha)$ and $P_2(\alpha)$ for finite ordinals.

We will proceed by a second induction on $\alpha$.

\textbf{Second Induction Basis ($\alpha=0$)}
We have $N_0 = \langle\emptyset,\emptyset\rangle$. As this interpretation does not assign the value $\mathit{true}$ or the value $\mathit{false}$ to any atom, both properties $P_1(0)$ and $P_2(0)$ hold vacuously.

\textbf{Second Induction Step ($\alpha+1$)}
We first show $P_1(\alpha+1)$. We have that $N_{\alpha+1}=\langle\Psi^{\uparrow \omega}_{N_\alpha},\mathcal{B}_{\alpha+1}-\Psi^{\uparrow \omega}_{N_\alpha}\rangle$; observe that $\vgr{\mathsf{A}\theta}{\Psi^{\uparrow \omega}_{N_\alpha}}=\mathit{true}$ if and only if there exists some $n<\omega$ for which $\vgr{\mathsf{A}\theta}{\Psi^{\uparrow n}_{N_\alpha}}=\mathit{true}$. Therefore, in order to prove $P_1(\alpha+1)$, we first need to perform a third induction on $n$ and prove the following property:
\begin{description}
\item[$P'_1(\alpha+1,n)$:] If $\vgr{\mathsf{A}\theta}{\Psi_{N_\alpha}^{\uparrow n}}=\mathit{true}$ then $\vgr{\mathsf{A}\theta'}{\pgp}=\mathit{true}$.
\end{description}

\textbf{Third Induction Basis ($n=0$)}
Property $P'_1(\alpha+1,0)$ holds vacuously, since $\Psi_{M_\alpha}^{\uparrow 0}=\emptyset$, i.e. it does not assign the value $\mathit{true}$ to any atom.

\textbf{Third Induction Step ($n+1$)}
We now show property $P_1'(\alpha+1,n+1)$, assuming that $P_1'(\alpha+1,n)$ holds. If $\vgr{\mathsf{A}\theta}{\Psi_{N_\alpha}^{\uparrow (n+1)}}=\mathit{true}$, then there exists a clause $\mathsf{A}\theta\leftarrow \mathsf{L}_1,\ldots, \mathsf{L}_k$ in $\mathsf{Gr(P)}$ such that, for each $i\leq k$, either $\vgr{\mathsf{L}_i}{{N_\alpha}}=\mathit{true}$ or $\mathsf{L}_i$ is an atom and $\vgr{\mathsf{L}_i}{\Psi_{N_\alpha}^{\uparrow n}}=\mathit{true}$.
This clause is a ground instance of a clause $\mathsf{p}\,\mathsf{V}_1\,\cdots\,\mathsf{V}_m \leftarrow \mathsf{K}_1,\ldots, \mathsf{K}_k$ in the higher-order program and there exists a substitution $\theta''$, such that $(\mathsf{p}\,\mathsf{V}_1\,\cdots\,\mathsf{V}_m)\theta''=\mathsf{A}$ and, for any variable $\mathsf{V}\not\in \{\mathsf{V}_1,\ldots,\mathsf{V}_m\}$ appearing in the body of the clause, $\theta''(\mathsf{V})$ is an appropriate ground term, so that $\mathsf{L}_i = \mathsf{K}_i\theta''\theta$ for all $i\leq k$. Observe that the variables appearing in the clause $(\mathsf{p}\,\mathsf{V}_1\,\cdots\,\mathsf{V}_m)\theta'' \leftarrow \mathsf{K}_1\theta'',\ldots, \mathsf{K}_k\theta''$ are exactly the variables appearing in $\mathsf{A}$ and they are all of types simpler than $\pi$. We distinguish the following cases for each $\mathsf{K}_i\theta''$, $i\leq k$:
\begin{enumerate}
\item \emph{$\mathsf{K}_i\theta''$ is of the form $(\mathsf{E}_1 \approx \mathsf{E}_2)$:} As remarked in~\ref{app:min_mod}, an expression of the form $(\mathsf{E}_1 \approx \mathsf{E}_2)$ has the same value in any interpretation. If $\vgr{\mathsf{K}_i\theta''\theta}{N_\alpha}=\vgr{\mathsf{K}_i\theta''\theta}{\Psi_{N_\alpha}^{\uparrow n}}=\mathit{true}$, by definition we have $\mathsf{E}_1\theta = \mathsf{E}_2\theta$. Since $\mathsf{E}_1$ and $\mathsf{E}_2$ are expressions of type $\iota$, all variables in $\mathsf{E}_1$ and $\mathsf{E}_2$ are also of type $\iota$ and, because $\exeq[\pgp,\iota]$ is defined as equality, we will have $\mathsf{E}_1\theta = \mathsf{E}_1\theta'$ and $\mathsf{E}_2\theta=\mathsf{E}_2\theta'$. Therefore $\mathsf{E}_1\theta' = \mathsf{E}_2\theta'$ and $\vgr{\mathsf{K}_i\theta''\theta'}{\pgp}=\mathit{true}$ will also hold.
\item \emph{$\mathsf{K}_i\theta''$ is an atom and starts with a predicate constant:} As we observed, the variables appearing in $\mathsf{K}_i\theta''$ are of types simpler than $\pi$. Because $\mathsf{K}_i\theta''\theta$ is an atom, either $\vgr{\mathsf{L}_i}{{N_\alpha}} = \vgr{\mathsf{K}_i\theta''\theta}{N_\alpha} = \mathit{true}$ or $\vgr{\mathsf{L}_i}{\Psi_{N_\alpha}^{\uparrow n}} =\vgr{\mathsf{K}_i\theta''\theta}{\Psi_{N_\alpha}^{\uparrow n}} = \mathit{true}$ may hold. In the former case, by the second induction hypothesis we can apply property $P_1(\alpha)$ and it follows that $\pgp(\mathsf{K}_i\theta''\theta')=\mathit{true}$. Similarly, in the latter case, the same conclusion can be reached by the third induction hypothesis and property $P_1'(\alpha+1,n)$.
\item \emph{$\mathsf{K}_i\theta''$ is an atom and starts with a predicate variable:} As in the previous case, it may be $\vgr{\mathsf{L}_i}{{N_\alpha}} = \vgr{\mathsf{K}_i\theta''\theta}{N_\alpha} = \mathit{true}$ or $\vgr{\mathsf{L}_i}{\Psi_{N_\alpha}^{\uparrow n}} =\vgr{\mathsf{K}_i\theta''\theta}{\Psi_{N_\alpha}^{\uparrow n}} = \mathit{true}$. Let $\mathsf{K}_i\theta''=\mathsf{V}\,\mathsf{E}_1\,\cdots\,\mathsf{E}_{m'}$ for some $\mathsf{V}\in vars(\mathsf{A})$. Then $\mathsf{B}=\theta(\mathsf{V})\,\mathsf{E}_1\,\cdots\,\mathsf{E}_{m'}$ is an atom that begins with a predicate constant and, by $vars(\mathsf{K}_i\theta'') \subseteq vars(\mathsf{A})$, all of the variables of $\mathsf{B}$ are of types simpler than $\pi$. Hence, by the second induction hypothesis, $\mathsf{B}$ satisfies property $P_1(\alpha)$ and if $\vgr{\mathsf{K}_i\theta''\theta}{N_\alpha} =\vgr{\mathsf{B}\theta}{N_\alpha} = \mathit{true}$ then it follows that $\pgp(\mathsf{B}\theta')=\mathit{true}$~(1).
Similarly, by the third induction hypothesis, $\mathsf{B}$ also satisfies property $P'_1(\alpha+1,n)$, so if $\vgr{\mathsf{K}_i\theta''\theta}{\Psi_{N_\alpha}^{\uparrow n}} =\vgr{\mathsf{B}\theta}{\Psi_{N_\alpha}^{\uparrow n}} = \mathit{true}$, then the same conclusion, that $\pgp(\mathsf{B}\theta')=\mathit{true}$~(1), is reached again. Observe that the types of all arguments of $\theta(\mathsf{V})$, i.e. the types of $\mathsf{E}_j\theta'$ for all $j\leq m'$, are simpler than the type of  $\mathsf{V}$ and consequently, since $\mathsf{V} \in vars(\mathsf{A})$, simpler than $\pi$. For each $j\leq m'$, let $\rho_j$ be the type of $\mathsf{E}_j$ and let $\rho$ be the type of $\mathsf{V}$; by the first induction hypothesis, $\mathsf{E}_j\theta'\exeq[\pgp,\rho_j] \mathsf{E}_j\theta'$. Moreover, by assumption we have that $\theta(\mathsf{V}) \exeq[\pgp,\rho] \theta'(\mathsf{V})$. Then, by definition and by (1) $\vgr{\theta(\mathsf{V})\,\mathsf{E}_1\theta'\,\cdots\,\mathsf{E}_{m'}\theta'}{\pgp}=\vgr{\theta'(\mathsf{V})\,\mathsf{E}_1\theta'\,\cdots\,\mathsf{E}_{m'}\theta'}{\pgp}=\vgr{\mathsf{K}_i\theta''\theta'}{\pgp}=\mathit{true}$.
\item \emph{$\mathsf{K}_i\theta''$ is a negative literal and its atom starts with a predicate constant:}
Let $\mathsf{K}_i\theta''$ be of the form $\pnot \mathsf{B}$, where $\mathsf{B}$ is an atom that starts with a predicate constant. It is $\vgr{\pnot\mathsf{B}\theta}{N_{\alpha}}=\vgr{\mathsf{K}_i\theta''\theta}{N_{\alpha}}=\vgr{\mathsf{L}_i}{N_{\alpha}}= \mathit{true}$ and therefore $\vgr{\mathsf{B}\theta}{N_{\alpha}}=\mathit{false}$. Moreover, by $vars(\mathsf{K}_i\theta'') \subseteq vars(\mathsf{A})$, all the variables of $\mathsf{B}$ are of types simpler than $\pi$, so we can apply the second induction hypothesis, in particular property $P_2(\alpha)$, to $\mathsf{B}$ and conclude that $\vgr{\mathsf{B}\theta'}{\pgp}= \mathit{false}$. Then $\vgr{\pnot \mathsf{B}\theta'}{\pgp}= \vgr{\mathsf{K}_i\theta''\theta'}{\pgp}=\mathit{true}$.
\item \emph{$\mathsf{K}_i\theta''$ is a negative literal and its atom starts with a predicate variable:}
Let $\mathsf{K}_i\theta''=\pnot(\mathsf{V}\,\mathsf{E}_1\,\cdots\,\mathsf{E}_{m'})$ for some $\mathsf{V}\in vars(\mathsf{A})$. Then $\mathsf{B}=\theta(\mathsf{V})\,\mathsf{E}_1\,\cdots\,\mathsf{E}_{m'}$ is an atom that begins with a predicate constant and, by $vars(\mathsf{K}_i\theta'') \subseteq vars(\mathsf{A})$, all of the variables of $\mathsf{B}$ are of types simpler than $\pi$. Also, $\vgr{\pnot\mathsf{B}\theta}{N_{\alpha}} =\vgr{\mathsf{K}_i\theta''\theta}{N_{\alpha}}=\vgr{\mathsf{L}_i}{N_{\alpha}}=\mathit{true}$ and therefore $\vgr{\mathsf{B}\theta}{N_{\alpha}}=\mathit{false}$. Hence, by the second induction hypothesis and in particular property $P_2(\alpha)$, it follows that $\pgp(\mathsf{B}\theta')=\vgr{\theta(\mathsf{V})\,\mathsf{E}_1\theta'\,\cdots\,\mathsf{E}_{m'}\theta'}{\pgp}=\mathit{false}$~(1). Observe that the types of all arguments of $\theta(\mathsf{V})$, i.e. the types of $\mathsf{E}_j\theta'$ for all $j\leq m'$, are simpler than the type of  $\mathsf{V}$ and consequently, since $\mathsf{V} \in vars(\mathsf{A})$, simpler than $\pi$. For each $j\leq m'$, let $\rho_j$ be the type of $\mathsf{E}_j$ and let $\rho$ be the type of $\mathsf{V}$; by the first induction hypothesis, $\mathsf{E}_j\theta'\exeq[\pgp,\rho_j] \mathsf{E}_j\theta'$. Moreover, by assumption we have that $\theta(\mathsf{V}) \exeq[\pgp,\rho] \theta'(\mathsf{V})$. Then, by definition and by (1), $\vgr{\theta(\mathsf{V})\,\mathsf{E}_1\theta'\,\cdots\,\mathsf{E}_{m'}\theta'}{\pgp}=\vgr{\theta'(\mathsf{V})\,\mathsf{E}_1\theta'\,\cdots\,\mathsf{E}_{m'}\theta'}{\pgp}=\mathit{false}$. Obviously, this makes $\vgr{\pnot (\theta'(\mathsf{V})\,\mathsf{E}_1\theta'\,\cdots\,\mathsf{E}_{m'}\theta')}{\pgp}=\vgr{\mathsf{K}_i\theta''\theta'}{\pgp}=\mathit{true}$.
\end{enumerate}
We have shown that, for each $i\leq k$, $\vgr{\mathsf{K}_i\theta''\theta'}{\pgp}=\mathit{true}$. Since the clause $\mathsf{A}\theta' \leftarrow \mathsf{K}_1\theta''\theta',\ldots, \mathsf{K}_k\theta''\theta'$ is in $\mathsf{Gr(P)}$ and $\pgp$ is a model of $\mathsf{Gr(P)}$, we can conclude that $\vgr{\mathsf{A}\theta'}{\pgp}=\mathit{true}$.

This concludes the proof for $P'_1(\alpha+1,n)$. Notice that property $P'_1(\alpha+1,n)$ immediately implies property $P_1(\alpha+1)$: as mentioned before, $\vgr{\mathsf{A}\theta}{N_{\alpha+1}}=\vgr{\mathsf{A}\theta}{\Psi^{\uparrow \omega}_{N_\alpha}}=\mathit{true}$ if and only if there exists some $n<\omega$ for which $\vgr{\mathsf{A}\theta}{\Psi^{\uparrow n}_{N_\alpha}}=\mathit{true}$ and then $\vgr{\mathsf{A}\theta'}{\pgp}=\mathit{true}$ follows from property $P'_1(\alpha+1,n)$. It remains to prove property $P_2(\alpha+1)$. Observe that the atoms $\mathsf{A}\theta$ and $\mathsf{A}\theta'$ both start with the same predicate constant $\mathsf{p}$ and recall that we have chosen a local stratification for $\mathsf{Gr(P)}$, such that $stratum(\mathsf{A}\theta)=stratum(\mathsf{A}\theta')=stratum(\mathsf{p})$. Moreover, we have that $N_{\alpha+1}=\langle\Psi^{\uparrow \omega}_{N_\alpha},\mathcal{B}_{\alpha+1}-\Psi^{\uparrow \omega}_{N_\alpha}\rangle$, so if $\vgr{\mathsf{A}\theta}{N_{\alpha+1}}=\mathit{false}$, it follows that $\mathsf{A}\theta\in\mathcal{B}_{\alpha+1}$. Because $stratum(\mathsf{A}\theta)=stratum(\mathsf{A}\theta')$, it must also be $\mathsf{A}\theta'\in\mathcal{B}_{\alpha+1}$,  which implies that $\vgr{\mathsf{A}\theta'}{N_{\alpha+1}}$ can be either $\mathit{true}$ (if $\mathsf{A}\theta'\in \Psi^{\uparrow \omega}_{N_\alpha}$) or $\mathit{false}$ (if $\mathsf{A}\theta'\not\in \Psi^{\uparrow \omega}_{N_\alpha}$), but not $\unk$. For the sake of contradiction, assume that $\vgr{\mathsf{A}\theta'}{N_{\alpha+1}}=\mathit{true}$. As the relations $\exeq[\pgp,\rho_i]$ are symmetric, $\theta$ and $\theta'$ are interchangeable, so property $P_1(\alpha+1)$ applies and yields $\vgr{\mathsf{A}\theta}{\pgp}=\mathit{true}$. Because (by Theorem~\ref{thr:perf_model_char}) $N_{\alpha+1}\fleq\pgp$, this contradicts our initial assumption that $\vgr{\mathsf{A}\theta}{N_{\alpha+1}}=\mathit{false}$. Therefore, it must be $\vgr{\mathsf{A}\theta'}{N_{\alpha+1}}=\mathit{false}$ and so, again by $N_{\alpha+1}\fleq\pgp$, it follows that $\vgr{\mathsf{A}\theta'}{\pgp}=\mathit{false}$.
\end{proof}
\end{document}